\documentclass[11pt]{article}

\usepackage{amsmath}
\usepackage{fancyhdr}
\usepackage{amssymb}
\usepackage{amsthm}
\usepackage{graphicx}
\usepackage{varioref}
\usepackage{verbatim} 
\usepackage{multicol}
\usepackage{enumerate}
\usepackage[normalem]{ulem}
\usepackage{caption}
\usepackage{subcaption}
\usepackage[T1]{fontenc}
\usepackage[margin=1in]{geometry}
\usepackage{float}
\usepackage{upgreek}
\usepackage{mathtools}
\usepackage[linesnumbered,ruled,vlined]{algorithm2e}

\usepackage{mathrsfs}

\usepackage{url}
\usepackage{xspace}
\usepackage{thm-restate}

\usepackage{lmodern} 
\usepackage[T1]{fontenc}

\usepackage{microtype}

\usepackage{tikz}
\usetikzlibrary{calc}

\usepackage{epstopdf}

\usepackage[numbers,sort&compress]{natbib}
\usepackage[colorlinks = {true}, pdftex,pdftitle={Leveraging Unknown Structure in Quantum Query Algorithms}]{hyperref}

\usepackage{xcolor}
\definecolor{darkred}  {rgb}{0.5,0,0}
\definecolor{darkblue} {rgb}{0,0,0.5}
\definecolor{darkgreen}{rgb}{0,0.5,0}
\hypersetup{
  urlcolor   = blue,         
  linkcolor  = darkblue,     
  citecolor  = darkgreen,    
  filecolor  = darkred       
}

\usepackage{cleveref}
\crefname{lemma}{Lemma}{Lemmas}
\crefname{proposition}{Proposition}{Propositions}
\crefname{definition}{Definition}{Definitions}
\crefname{theorem}{Theorem}{Theorems}
\crefname{conjecture}{Conjecture}{Conjectures}
\crefname{corollary}{Corollary}{Corollaries}
\crefname{section}{Section}{Sections}
\crefname{appendix}{Appendix}{Appendices}
\crefname{figure}{Fig.}{Figs.}
\crefname{algorithm}{Alg.}{Algs.}
\crefname{equation}{Eq.}{Eqs.}
\crefname{table}{Table}{Tables}
\crefname{claim}{Claim}{Claims}

\newtheorem{theorem}{Theorem}
\newtheorem{lemma}[theorem]{Lemma}

\newtheorem{definition}[theorem]{Definition}
\newtheorem{corollary}[theorem]{Corollary}

\newtheorem*{conjecture*}{Conjecture}

\theoremstyle{definition}

\usepackage{authblk}

\newcommand{\ket}[1]{|#1\rangle}
\newcommand{\bra}[1]{\langle#1|}
\newcommand{\proj}[1]{|#1\rangle\!\langle#1|}

\DeclareMathOperator*{\argmin}{arg\,min}

\DeclareMathAlphabet{\matheu}{U}{eus}{m}{n}

\newcommand{\sop}[1]{{\mathcal #1}}

\newcommand{\braket}[2]{\langle{#1}|{#2}\rangle}
\newcommand{\ketbra}[2]{|{#1}\rangle\!\langle{#2}|}

\newcommand{\iner}[1]{\langle{#1}|{#1}\rangle}

\newcommand{\Lal}[1]{{\Lambda^{\!{#1}}}}
\newcommand{\U}[3]{{U({#1},{#2},{#3})}}
\newcommand{\w}[2]{{w({#1},{#2})}}
\newcommand{\Aa}[1]{{\tilde{A}^{#1}}}
\newcommand{\PtHx}{\Pi_{x}}
\newcommand{\PHx}{\Pi_{{H}(x)}}
\newcommand{\Lalp}{{\Uplambda^{\!\alpha,\hat{\varepsilon}}}}
\newcommand{\PHxp}{\Uppi_{x}}
\newcommand{\Up}[4]{{\sop U({#1},{#2},{#3},{#4})}}

\begin{document}

\title{Leveraging Unknown Structure in Quantum Query Algorithms}
\author[1]{Noel T. Anderson}
\author[1]{Jay-U Chung}
\author[1]{Shelby Kimmel\footnote{skimmel@middlebury.edu}}
\affil[1]{Middlebury College, Department of Computer Science\\ Middlebury, VT, USA}

\date{}

\maketitle

\abstract{
Quantum span program algorithms for function evaluation commonly have reduced query complexity when promised that the input has a certain structure. We design a modified span program algorithm to show these speed-ups persist even without having a promise ahead of time, and we extend this approach to the more general problem of state conversion. For example, there is a span program algorithm that decides whether two vertices are connected in an $n$-vertex graph with $O(n^{3/2})$ queries in general, but with $O(\sqrt{k}n)$ queries if promised that, if there is a path, there is one with at most $k$ edges. Our algorithm uses $\tilde{O}(\sqrt{k}n)$ queries to solve this problem if there is a path with at most $k$ edges, without knowing $k$ ahead of time.}
\section{Introduction}\label{sec:intro}
Quantum algorithms often yield speed-ups when given a promise on
the input. For example, if we know that there are $M$ marked items out of $N$, or no marked items at all, then Grover's search can be run in time $O(\sqrt{N/M})$,
rather than $O(\sqrt{N})$, the worst case runtime with
a single marked item \cite{aharonovQuantumComputation1999}.

In the case of Grover's algorithm, a series of results
\cite{boyerTightBoundsQuantum1998,brassardQuantumAmplitudeAmplification2000,brassardQuantumCounting1998,yoderFixedPointQuantumSearch2014} removed the promise; if there are $M$ marked items, there is a quantum search algorithm that runs in $O(\sqrt{N/M})$ time, even without knowing the number of marked items ahead of time.
Most relevant for our
work, several of these algorithms involve iteratively running Grover's
search with exponentially growing runtimes \cite{boyerTightBoundsQuantum1998,brassardQuantumAmplitudeAmplification2000} until a marked item is found.

Grover's algorithm was one of the first quantum query algorithms discovered \cite{groverQuantumMechanicsHelps1997a}.
Since that time, a much more general framework for quantum query algorithms
has been laid out. Span programs, and more generally, the dual
of the general adversary bound and the filtered $\gamma_2$ norm, are frameworks for creating optimal query
algorithms for function decision problems
\cite{reichardtSpanProgramsQuantum2009,reichardtReflectionsQuantumQuery2011}
and nearly optimal algorithms for state conversion problems, in which the goal is to generate a quantum state based on an oracle and an input state \cite{leeQuantumQueryComplexity2011}. Moreover, these frameworks are also useful in practice \cite{beigiQuantumSpeedupBased2019,belovsSpanProgramsQuantum2012,belovsSpanProgramsFunctions2012,belovsTightQuantumLower2020a,cadeTimeSpaceEfficient2018,delorenzoApplicationsQuantumAlgorithm2019a}.

For many span program algorithms, analogous to multiple marked items in Grover's search, there are features
which, if promised they exist, allow for improvement over the worst case query complexity. For example, the span program version of Grover's algorithm also runs faster with the promise of multiple marked items. Another well-studied example is the span program algorithm for 
$st$-connectivity. This algorithm uses
$O\left(n^{3/2}\right)$ queries on an $n$-vertex graph. However, if promised ahead of time that if there is
a path, it has length at most $k$, then the problem can be solved with
$O(\sqrt{k}n)$ queries \cite{belovsSpanProgramsQuantum2012}.

Our contribution is to speed up the generic span program and state conversion algorithms in the case that some speed-up inducing structure (such as multiple marked items, or a short path) is present, even without the promise on the structure ahead of time. One might expect that it
is trivial to remove the promise. Surely if an algorithm produces a correct result in shorter time when promised a property is present, then it should also produce a correct result in a shorter time 
without the promise if the property still holds? While this is true and these algorithms always output a result, even if run for a short time, the problem is that they don't produce a flag of completion and their output can not be easily verified. Without a flag of completion or a promise of structure, it is impossible to be confident that the result is correct. Span program and state conversion
algorithms differ from Grover's algorithm in their lack of a flag; in Grover's algorithm one can test
with a single query whether the output is a marked item, thus flagging that the output of the algorithm is correct, and that
the algorithm has run for a sufficiently long time. We note that when span program algorithms previously have claimed a speed-up with
structure, they always included a promise, or they give the
disclaimer that running the algorithm will give an incorrect result with
high probability if the promise is not known ahead of time to be satisfied, e.g. \cite[App. C.3] {cadeTimeSpaceEfficient2018}.

To overcome this challenge, we use an approach that is
similar to the iterative modifications to Grover's algorithm; we run subroutines for exponentially increasing times, and we have novel ways to flag when the computation should halt. Our algorithms match the performance of algorithms with an optimal promise, up to logarithmic factors.


Based on prior work showing the presence of a speed-up for span program algorithms with a promise, our work immediately provides analogous speed-ups without the promise:
\begin{itemize}
\item For undirected $st$-connectivity described above, our algorithm determines whether there is a path from $s$ to
$t$ in an $n$-vertex graph with $\tilde{O}(\sqrt{k}n)$ queries if there is a path of length $k$,
 and if there is no path,
the algorithm uses $\tilde{O}(\sqrt{nc})$ queries, where $c$ is the size of
the smallest cut between $s$ and $t$. In either case, $k$ and $c$ need not be known ahead of time.
\item For an $n$-vertex undirected graph, we can determine if it is connected in $\tilde{O}(n\sqrt{R})$ queries, where $R$ is the average
effective resistance and we can determine the graph is not connected in $\tilde{O}(\sqrt{n^{3}/\kappa})$ queries, where $\kappa$ is the number of components. These query complexities hold without knowing $R$ or $\kappa$ ahead of time. See Ref.
\cite{jarretQuantumAlgorithmsConnectivity2018} for the promise version of this problem.
\item For cycle detection on an $n$-vertex undirected graph,
whose promise version was analyzed in Ref.
\cite{delorenzoApplicationsQuantumAlgorithm2019a}, if the circuit rank is $C$, then our algorithm will decide
there is a cycle in $\tilde{O}(\sqrt{n^{3}/C})$ queries, while if there is
no cycle and at most $\mu$ edges, the algorithm will decide there is no cycle
in $\tilde{O}(\mu\sqrt{n})$ queries. This result holds
without knowing $C$ or $\mu$ ahead of time.
\item We can decide the winner of the NAND-tree \cite{farhi2007quantum} and similar games, whose promise version was analyzed in \cite{zhan2012super,kimmel2011quantum}, in $\tilde{O}(b^k)$ queries for a constant $b$ that depends on the details of the game, if the game satisfies the $k$-fault condition, a bound on the number of consequential decisions a player must make through the course of the game. The parameter $k$ need not be known ahead of time.
\end{itemize}


\subsection{Directions for Future Work}

In the original state conversion
algorithm, to achieve an error of $\varepsilon$ in the output state (by some
metric), the query complexity scales as
$O\left(\varepsilon^{-2}\right)$. However, in our result, the query complexity scales as $O\left(\varepsilon^{-5}\right)$. In cases where accuracy must
scale with the input size, this error term could overwhelm any
advantage from our approach, and so it would be beneficial to improve this error scaling.

By iterating with exponentially increasing runtimes, we incur a logarithmic factor to manage errors. However, the 
fixed-point method for Grover's algorithm \cite{yoderFixedPointQuantumSearch2014}
avoids this overhead. Perhaps fixed-point techniques could improve the performance of our algorithms. 

Our result is similar to that of Ito and Jeffery
\cite{itoApproximateSpanPrograms2019}, which estimates structure size (e.g. the number of marked items in the case of search) with queries proportional to the square root of the structure size. While there are similarities between our approaches, neither result seems to directly imply the other.
Better understanding the relationship between these strategies could lead to improved
algorithms for determining properties of input structure for both
span programs and state generation problems.

Finally, these algorithms would likely provide opportunities for proving significant quantum to classical speed-ups in average-case query complexity. Exponential and super-polynomial speed-ups for average case query complexity have been shown when the quantum algorithm can alert the user that the algorithm has terminated early \cite{ambainis2001average,montanaro2010quantum}. Our work significantly expands the set of algorithms for which such an early termination flag exists. 

\section{Preliminaries}\label{sec:Prelims}

\textbf{Basic Notation:} Let $[n]$ represent $\{1,2,\dots,n\}$, and $\log$ denotes base 2 logarithm. We denote a linear operator from the space $V$ to the space $U$ as $\sop L(V,U)$. We use $I$ for the identity operator. (It will be clear from context which space $I$ acts on.) Given a projection $\Pi$, its complement
is $\overline{\Pi}=(I-\Pi).$ For a matrix $M$, by $M_{xy}$ or $(M)_{xy}$, we denote the element in the $x$th row and $y$th column of $M$. By $\tilde{O}$, we denote big-O notation that ignores log factors. The $l_2$-norm of a vector $\ket{v}$ is denoted by $\|\ket{v}\|$. For any unitary $U$, let $P_\Theta(U)$ be the projection onto the eigenvectors
of $U$ with phase at most $\Theta$. That is,
$P_\Theta(U)$ is the projection onto $\textrm{span}\{\ket{u}:U\ket{u}=e^{i\theta}\ket{u}\textrm{ with
}|\theta|\leq\Theta\}$.

\subsection{Quantum Algorithmic Building Blocks}

In this paper, we consider quantum query algorithms, in which one has access to a unitary $O_x$, called the oracle, which encodes a string $x\in X$ for $X\subseteq [q]^n$ and some parameter $q\geq 2$. The oracle $O_x$ acts on the Hilbert space $\mathbb{C}^n\otimes \mathbb{C}^q$ as
$O_x\ket{i}\ket{b}=\ket{i}\ket{x_i+b\textrm{ mod }q}$, where $x_i\in[q]$ is the $i$th
element of $x$.

Given $O_x$, and with no knowledge of $x$ ahead of time, except that $x\in X$,
we would like to perform a computation that depends on $x$. The query complexity is the minimum number of uses of the oracle required such that for all $x\in X$, the computation is successful with some desired probability of success.

Several of our key algorithmic subroutines are based around a parallelized
version of phase estimation, as described in Ref.
\cite{magniezSearchQuantumWalk2011}. The idea of this parallel phase
estimation algorithm is as follows: for a unitary $U$, a precision $\Theta>0$,
and an accuracy $\epsilon>0$, the circuit $D(U)$ implements
$O(\log\frac{1}{\epsilon})$ copies of the phase estimation circuit on $U$,
each to precision $O(\Theta)$, that all measure the phase of a single state on
the same input register. If $U$ acts on a Hilbert Space $\sop H$, then $D(U)$
acts on the space $\sop H_A\otimes((\mathbb{C}^{2})^{\otimes b})_B$ for
$b=O\left(\log\frac{1}{\Theta}\log\frac{1}{\epsilon}\right)$, where we've used
$A$ to label the register that stores the input state, and $B$ to label the
registers that store the results of the parallel phase estimations. $D(U)$ is
uniformly constructed in $\epsilon$ and $\Theta$, and its structure is
independent of $U$.

The circuit $D(U)$ can be used for Phase Checking: checking if a state
$\ket{\psi}$ is close to an eigenvector of $U$ that has eigenvalue close to
$1$. This is related to the probability of outcome $\ket{0}_B$ after
$D(U)$ is applied to the state $\ket{\psi}_A\ket{0}_B$. To characterize this
probability, we define $\Pi_0(U)$ to be the orthogonal projection onto the
subspace of $\sop H_A\otimes((\mathbb{C}^{2})^{\otimes b})_B$ that $D(U)$ maps
to states with $\ket{0}_B$ in the $B$ register.  That is,
$\Pi_0(U)=D(U)^\dagger \left(I_A\otimes\proj{0}_B\right) D(U).$ (Since
$\Pi_0(U)$ depends on the choice of $\Theta$ and $\epsilon$ used in $D(U)$,
those values must be specified, if not clear from context, when discussing
$\Pi_0(U)$.)  We now summarize relevant prior results for Phase Checking
(see
\cite{kitaevQuantumMeasurementsAbelian1995,cleveQuantumAlgorithmsRevisited1998,magniezSearchQuantumWalk2011}
for detailed analysis) in \cref{lem:phase_det}:

\begin{lemma}[Phase Checking]
Let $U$ be a unitary on a Hilbert Space $\sop H$, and let $\Theta,\epsilon>0$. We call $\Theta$ the precision and $\epsilon$ the accuracy. Then there is a circuit $D(U)$ that acts on the space $\sop H_A\otimes
((\mathbb{C}^{2})^{\otimes b})_B$ for $b=O\left(\log\frac{1}{\Theta}\log\frac{1}{\epsilon}\right)$, and that
uses $O\left(\frac{1}{\Theta}\log\frac{1}{\epsilon}\right)$ calls to
control-$U$.  Then for any state $\ket{\psi}\in \sop H$
\begin{itemize}
\item $ \|P_0(U)\ket{\psi}\|^2\leq \|\Pi_0(U)\left(\ket{\psi}_A\ket{0}_B\right)\|^2\leq \|P_\Theta(U)\ket{\psi}\|^2+\epsilon,$ and
\item $\|\Pi_0(U)\left(\overline{P}_\Theta(U)\ket{\psi}\right)_A\ket{0}_B\|^2\leq \epsilon.$
\end{itemize}
\label{lem:phase_det}
\end{lemma}

It is possible to modify the Phase Checking circuit by implementing $D(U)$ as described above, applying a
$-1$ phase to the $A$ register if the $B$ register is \textit{not} in the
state $\ket{0}_B$, and then implementing $D(U)^\dagger$. We call this circuit Phase Reflection\footnote{In Ref.
\cite{leeQuantumQueryComplexity2011}, this procedure is referred to as ``Phase
Detection,'' but since no measurement is made, and rather only a reflection is
applied, we thought renaming this protocol as ``Phase Reflection''
would be more descriptive and easier to distinguish from ``Phase Checking.'' We apologize for any confusion this may cause when
comparing to prior work!}
and denote it as $R(U).$ Note that $R(U)=\Pi_0(U)-\overline{\Pi}_0(U)$, where $R(U)$ and $\Pi_0(U)$ have the same implicit precision $\Theta$ and accuracy $\epsilon$. The following lemma summarizes prior results on relevant properties of Phase Reflection.
\begin{lemma} [Phase Reflection \cite{magniezSearchQuantumWalk2011,leeQuantumQueryComplexity2011}]
Let $U$ be a unitary on a Hilbert Space $\sop H$, and let
$\Theta,\epsilon>0$. We call $\Theta$ the precision and $\epsilon$ the accuracy. Then there is a circuit $R(U)$ that acts on the space
$\sop H_A\otimes ((\mathbb{C}^{2})^{\otimes b})_B$ for
$b=O(\log(1/\Theta)\log(1/\epsilon))$, and that uses
$O\left(\frac{1}{\Theta}\log\frac{1}{\epsilon}\right)$ calls to control-$U$
and control-$U^\dagger$ such that
\begin{itemize}
\item $R(U)(P_0\ket{\psi})\ket{0}_B=(P_0\ket{\psi})_A\ket{0}_B$, and
\item $\|(R(U)+I)(\overline{P}_\Theta\ket{\psi})_A\ket{0}_B\|<\epsilon $.
\end{itemize}
Furthermore, $R(U)$ is uniformly constructed in $\epsilon$ and $\Theta$, with structure independent of $U$. 
\label{lem:phase_refl}
\end{lemma}

Finally, we will use Amplitude Estimation \cite{brassardQuantumAmplitudeAmplification2000} and the effective spectral gap lemma \cite{leeQuantumQueryComplexity2011}:

\begin{lemma} [Amplitude Estimation \cite{brassardQuantumAmplitudeAmplification2000}] Let $\delta>0$, and let $\mathcal{A}$ be a quantum circuit such that $\mathcal{A}\ket{\psi}=\alpha_0\ket{0}\ket{\psi_0}+\alpha_1\ket{1}\ket{\psi_1}$. Then there is an algorithm that estimates $|\alpha_0|^2$ to additive error $\delta$ with success probability at least $1-p$ using $O\left(\frac{1}{\delta p}\right)$ calls to $\mathcal{A}$ and $\mathcal{A}^\dagger$.
\label{lem:ampEst}
\end{lemma}

\begin{lemma}[Effective spectral gap lemma, \cite{leeQuantumQueryComplexity2011}]
Let $\Pi$ and $\Lambda$ be projections, and let $U=(2\Pi-I)(2\Lambda-I)$ be the
unitary that is the product of their associated reflections. If
$\Lambda\ket{w}=0$, then $\|P_\Theta(U) \Pi\ket{w}\|\leq \frac{\Theta}{2}\|\ket{w}\|.$
\label{spec_gap_lemm}
\end{lemma}


\subsection{Span Programs}\label{sec:spanIntro}

\begin{definition} [Span Program] A span program is a tuple $P=(H,V,\tau,A)$ on $[q]^n$ where
\begin{enumerate}
\item $H$ is a direct sum of finite-dimensional inner product spaces: $H=H_1\oplus H_2\cdots H_n\oplus H_\textrm{true}\oplus H_\textrm{false},$
and for $j\in [n]$ and $a\in[q]$, we have $H_{j,a}\subseteq H_j$, such that $\sum_{a=1}^qH_{j,a}=H_j$.
\item $V$ is a vector space
\item $\tau\in V$ is a target vector, and
\item $A\in\sop L(H,V)$.
\end{enumerate}
Given a string $x\in[q]^n$, we use $H(x)$ to denote the subspace $H_{1,x_1}\oplus\cdots \oplus H_{n,x_n}\oplus H_\textrm{true}$, and we denote by $\PHx$ the orthogonal projection onto the space $H(x)$.

\label{def:SP}
\end{definition}

We use \cref{def:SP} for span programs because it applies to both binary and non-binary inputs ($q\geq 2$). The definitions in Refs. \cite{belovsSpanProgramsQuantum2012,cadeTimeSpaceEfficient2018} only apply to binary inputs ($q=2$).

\begin{definition} [Positive Witness]
Given a span program $P=(H,V,\tau,A)$ on $[q]^m$ and $x\in[q]^m$, then $\ket{w}\in H(x)$ is a positive witness for $x$ in $P$ if $A\ket{w}=\tau$. If a positive witness exists for $x$, we define the witness size of $x$ in $P$ as 
\begin{equation}
\w{P}{x}\coloneqq\min\left\{\|\ket{w}\|^2:\ket{w}\in H(x) \textrm{ and } A\ket{w}=\tau \right\}.
\end{equation}
We say that $\ket{w}\in H(x)$ is an optimal witness for $x$ if $\|\ket{w}\|^2=\w{P}{x}$ and $A\ket{w}=\tau$.
\label{def:posWitness}
\end{definition}

\begin{definition} [Negative Witness]
Given a span program $P=(H,V,\tau,A)$ on $[q]^m$ and $x\in[q]^m$, then $\omega\in \sop L(V,\mathbb{R})$ is a negative witness for $x$ in $P$ if $\omega\tau=1$ and $\omega A \PHx=0$. If a negative witness exists for $x$, we define the witness size of $x$ in $P$ as
\begin{equation}
 \w{P}{x}\coloneqq\min\left\{\|\omega A\|^2:
\omega\in L(V,\mathbb{R}), \omega A \PHx=0, \textrm{ and } \omega\tau=1 \right\}.
\end{equation}
We say that $\omega$ is an optimal witness for $x$ if $\|\omega A\|^2=\w{P}{x}$, $\omega A \PHx=0,$ and $\omega\tau=1$.
\label{def:negWit}
\end{definition}
\noindent Each $x\in[q]^m$ has a positive or negative witness (but not both), so $\w{P}{x}$ is well defined.

We say that a span program $P$ decides the function
$f:X\subseteq[q]^n\rightarrow \{0,1\}$ if each $x\in f^{-1}(1)$ has a positive
witness in $P$, and each $x\in f^{-1}(0)$ has a negative witness in $P$. Then denote $W_+(P,f)=\max_{x\in f^{-1}(1)}\w{P}{x}$ and let $W_-(P,f)=\max_{x\in f^{-1}(0)}\w{P}{x}$.

Given a description of a span program that decides a function, one
can use it to design a quantum query algorithm that evaluates the same function. 
The query complexity of the quantum algorithm depends on $W_+(P,f)$ and $W_-(P,f)$:
\begin{theorem} [\cite{reichardtSpanProgramsQuantum2009,itoApproximateSpanPrograms2019}]
For $X\subseteq[q]^n$ and $f:X\rightarrow \{0,1\}$, let $P$ be a span program that decides $f$. Then there is a quantum algorithm that for any $x\in X$, evaluates $f(x)$ with bounded error, and uses $O\left(\sqrt{W_+(P,f)W_-(P,f)}\right)$ queries to the oracle $O_x$.
\label{thm:spEval}
\end{theorem}

Not only can any span program that decides a function $f$ be used to create a
quantum query algorithm that decides $f$, but there is always a span program
that creates an algorithm with asymptotically optimal query complexity
\cite{reichardtSpanProgramsQuantum2009,reichardtReflectionsQuantumQuery2011}.
Thus when designing quantum query algorithms for function decision problems,
it is sufficient to consider only span programs.

As noted in Ref.~\cite{itoApproximateSpanPrograms2019}, one can scale and
normalize a span program $P$ to create a new span program $P'$ such that:
\begin{itemize}
\item All positive and negative witnesses of $P'$ have value at least 1.
\item $\left|W_+(P',f)-W_-(P',f)\right|\leq 1$.
\item $W_+(P',f)W_-(P',f)=O(W_+(P,f)W_-(P,f))$
\end{itemize}
The first point is achieved by scaling the target vector, see
\cite[Definition 2.13]{itoApproximateSpanPrograms2019}, and the next points
are achieved by applying \cite[Theorem 2.14]{itoApproximateSpanPrograms2019}
with $\beta=(W_+(P,f)/W_-(P,f))^{1/4}.$ The final point ensures that
the query complexity of the algorithm produced by the scaled and normalized
span program is the same as the original span program.

Thus without loss of generality, we henceforward assume our span programs are scaled and normalized such that the maximum positive and negative witnesses have the same size, which we denote $W(P,f)$. When clear from context, we will drop the input parameters and refer to $W(P,f)$ as simply $W$.

Given this renormalization, we can restate \cref{thm:spEval} in a way that is more
conducive for comparison with our results:
\begin{corollary}\label{cor:spEval}
For $X\subseteq[q]^n$, let $P$ be a (scaled and normalized) span program that decides $f:X \rightarrow
\left\{0,1\right\}$ with witness size $W=\max_{\chi\in X}\w{P}{\chi}$.  Then there is a bounded error quantum algorithm that for every input $x\in X$ evaluates $f(x)$ and uses $O\left(W\right)$ queries to the oracle $O_x$.
\end{corollary}

We will also find it helpful to use a transformation that takes a span program
$P$ that decides a function $f:X\rightarrow\{0,1\}$ and creates a span program $P^\dagger$ that decides
$\neg f$, the negation of $f$, without increasing witness sizes. While such a transformation is known for Boolean span programs \cite{reichardtSpanProgramsQuantum2009}, we show it exists for the span programs of \cref{def:SP}. (The proof can be found in \cref{app:sec2}.)

\begin{restatable}{lemma}{SPduals}\label{lem:SP_duals}
Given a span program $P=(H,V,\tau,A)$ on $[q]^m$ that decides a function
$f:X\rightarrow \{0,1\}$ for $X\in [q]^m$, there exists a span program
$P^\dagger=(H',V',\tau',A')$ such that  $\forall x\in X, \w{P}{x}\geq
\w{P^\dagger}{x}$, and $P^\dagger$ decides $\neg f$.
\end{restatable}


\subsection{State Conversion and Filtered \texorpdfstring{$\gamma_2$}{gamma-2} Norm}\label{sec:dualAdvIntro}

In the state conversion problem, for $X\subseteq [q]^n$, we are given descriptions of sets of pure
states $\{\ket{\rho_x}\}_{x\in X}$ and $\{\ket{\sigma_x}\}_{x\in X}$. (Moving forward, we will simply write $\{\ket{\rho_x}\}$ and $\{\ket{\sigma_x}\}$.) Then given access to an oracle for $x$, and the quantum state
$\ket{\rho_x}$, the goal is to create a state $\ket{\sigma_x'}$ such that
$\|\ket{\sigma_x'}-\ket{\sigma_x}\ket{0}\|\leq \varepsilon$. We call $\varepsilon$
the error of the state conversion procedure\footnote{We only consider what in Ref. \cite{leeQuantumQueryComplexity2011} is called the coherent state conversion problem. We are concerned with algorithms rather than lower bounds, and the algorithm we describe also applies to the less restrictive non-coherent problem.}.

The filtered $\gamma_2$ norm is used in designing quantum algorithms for state conversion:
\begin{definition} [Filtered $\gamma_2$ norm]\label{def:Gamma2}
Let $B$ be a matrix whose rows and columns are indexed by the elements of a set $X$. Let $Z=\{Z_1,\dots Z_n\}$ be a set of matrices whose rows and columns are indexed by the elements of $X$. Define $\gamma_2(B|Z)$ as
\begin{align}
\gamma_2(B|Z)=&\min_{\substack{m\in \mathbb{N}\\\ket{u_{xj}},\ket{v_{yj}}\in \mathbb{C}^m}}\max\left\{\max_{x\in X}\sum_j\|\ket{u_{xj}}\|^2,\max_{x\in X}\sum_j\|\ket{v_{yj}}\|^2\right\}\label{eq:filteredNormMa} \\
&\textrm{ s.t. }\forall x,y\in X, B_{xy}=\sum_{j=1}^n(Z_j)_{xy}\braket{u_{xj}}{v_{yj}}.\label{eq:filteredNorm}
\end{align}
\end{definition}

Let $\rho$ and $\sigma$ be the Gram matrices of the sets $\{\ket{\rho_x}\}$ and $\{\ket{\sigma_x}\}$, respectively. In other words, $\rho$ and $\sigma$ are matrices whose rows and columns are indexed by the elements of $X$ such that $\rho_{xy}=\braket{\rho_x}{\rho_y},$ and $\sigma_{xy}=\braket{\sigma_x}{\sigma_y}.$
Let $\Delta=\{\Delta_1, \dots,\Delta_n\}$ be a set of matrices whose rows and columns are indexed by the elements of $X$ such that the element in the $x$th row and $y$th column of $\Delta_j$ is $1$ if the $j$th element of $x$ and $y$ differ, and $0$ if they are the same: $(\Delta_j)_{xy}=1-\delta_{x_j,y_j}.$

Then the query complexity of state conversion is characterized as follows:

\begin{theorem} [\cite{leeQuantumQueryComplexity2011}]
Given $X\subseteq[q]^n$, and sets of states $\{\ket{\rho_x}\}$ and
$\{\ket{\sigma_x}\}$ for each $x\in X$, with respective Gram matrices $\rho$ and
$\sigma$, the query complexity of state conversion with error $\varepsilon$ is
$O\left(\gamma_2(\rho-\sigma|\Delta)\frac{\log(1/\varepsilon)}{\varepsilon^2}\right)$.
\label{thm:StateConv}
\end{theorem}

Any set of vectors that satisfies the constraints of \cref{eq:filteredNorm} for 
$\gamma_2(\rho-\sigma|\Delta)$ can be used to create an algorithm to solve the state conversion problem from
$\{\ket{\rho_x}\}$ and $\{\ket{\sigma_x}\}$, although it will not necessarily be optimal
in terms of query complexity; its query complexity will depend on the value of the objective function in \cref{eq:filteredNormMa} for that set of vectors \cite{leeQuantumQueryComplexity2011}. We call such a set of vectors a \textit{converting
vector set} from $\rho$ to $\sigma$:

\begin{definition} [Converting vector set]
We say a set of vectors $\sop{P}=\left(\{\ket{u_{xj}}\},\{\ket{v_{yj}}\}\right)$ converts $\rho$ to $\sigma$
if it satisfies the constraints of \cref{eq:filteredNorm} for $B=\rho-\sigma$, and $Z=\Delta$. We call such a $\sop P$ a \em{converting vector set} from $\rho$ to $\sigma$.
\end{definition}
\noindent Note that for a converting vector set $\sop{P}=\left(\{\ket{u_{xj}}\},\{\ket{v_{yj}}\}\right)$  from $\rho$ to $\sigma$, given the constraints of \cref{eq:filteredNorm}
with $Z_j$ set to $\Delta_j$, we have that $\forall x,y\in X,$
\begin{equation}\label{eq:cvsSum}
\rho_{xy}-\sigma_{xy}=\sum_{j:x_j\neq y_j}\braket{u_{xj}}{v_{jy}}.
\end{equation}

Analogous to witness sizes in span
programs, we define a notion of witness sizes for converting vector sets:
\begin{definition} [Converting vector set witness sizes]
Given a converting vector set
$\sop{P}=\left(\{\ket{u_{xj}}\},\{\ket{v_{yj}}\}\right)$, we define the
witness sizes of $\sop P$  as
\begin{align}
&w_+(\sop P,x)\coloneqq\sum_j\|\ket{u_{xj}}\|^2 & &\textrm{positive witness size for $x$ in $\sop P$}\nonumber\\
&w_-(\sop P,x)\coloneqq\sum_j\|\ket{v_{xj}}\|^2 & &\textrm{negative witness size for $x$ in $\sop P$}\nonumber\\
&W(\sop {P})\coloneqq\max_{x\in X}\left\{\max \{w_+(\sop P,x),\, w_-(\sop P,x)\}\right\} & &\textrm{witness size of $\sop P$}
\end{align}
\end{definition}
If $\sop P$ is a converting vector set from $\rho$ to $\sigma$, then
the value of the objective function in \cref{eq:filteredNormMa} equals $W(\sop P)$, and thus the query complexity of converting from $\rho$ to $\sigma$ using $\sop P$ is $O\left(W(\sop P)\frac{\log(1/\varepsilon)}{\varepsilon^2}\right)$. When clear from context, we refer to $W(\sop P)$ as $W$.

The following two lemmas (whose proofs can be found in \cref{app:sec2}), provide some useful transformations for converting vector sets:
\begin{restatable}{lemma}{complment}\label{lem:complement}
If $\sop{P}=\left(\{\ket{u_{xj}}\},\{\ket{v_{yj}}\}\right)$ converts $\rho$ to $\sigma$, then there is a complementary
converting vector set $P^C=\left(\{\ket{u^C_{xj}}\},\{\ket{v^C_{yj}}\}\right)$ that also converts $\rho$ to $\sigma$, such that for all $x\in X$ and for all $j\in [n]$, we have $w_+(\sop P,x)=w_-(\sop P^C,x)$, and $w_-(\sop P,x)=w_+(\sop P^C,x)$; the complement exchanges the values of the positive and negative witness sizes.
\end{restatable}

\begin{restatable}{lemma}{CVStransform}\label{lem:CVStransform1}
Given
$\sop{P}=\left(\{\ket{u_{xj}}\},\{\ket{v_{yj}}\}\right)$ that converts $\rho$ to $\sigma$,
there exists a normalized converting vector set
$\sop{P}'=\left(\{\ket{u'_{xj}}\},\{\ket{v'_{yj}}\}\right)$ from $\rho$ to $\sigma$ such
that $W(\sop P')\leq W(\sop P)$ and $\max_{x\in X} w_+(\sop P',x)=\max_{x\in X} w_-(\sop P',x)$.
\end{restatable}

Thus without loss of generality, we henceforward assume our converting vector
sets are normalized such that the maximum positive and negative witness sizes
are the same size.  In particular, given \cref{lem:CVStransform1}, and using the algorithm
for state conversion \cite{leeQuantumQueryComplexity2011}, we can restate
\cref{thm:StateConv} in a way that will be more conducive to comparing to our result:
\begin{corollary}\label{cor:stateConv}
Let $\sop{P}$ be a (normalized) converting
vector set from $\rho$ to $\sigma$ with witness size $W=\min\left\{\max_{\chi\in X}w_+(\sop P,\chi),\max_{\chi\in X}w_-(\sop P,\chi)\}\right\}$. Then there is quantum algorithm that on every input $x\in X$ converts $\ket{\rho_x}$ to $\ket{\sigma_x}$ with error $\varepsilon$ and uses 
$\tilde{O}\left(W/\varepsilon^2\right)$ queries to the oracle $O_x$.
\end{corollary}

We will make use of the unit vectors $\{\ket{\mu_i}\}_{i\in [q]}$ and $\{\ket{\nu_i}\}_{i\in [q]}$, which are commonly seen in dual adversary algorithms like state-conversion:
\begin{equation}\label{eq:mu_nu}
\ket{\mu_i}=-\alpha\ket{i}+\sqrt{\frac{1-\alpha^2}{q-1}}\sum_{j\neq i}\ket{j}, \qquad 
\ket{\nu_i}=\sqrt{1-\alpha^2}\ket{i}+\frac{\alpha}{\sqrt{q-1}}\sum_{j\neq i}\ket{j},
\end{equation}
where $\alpha=\sqrt{1/2-\sqrt{q-1}{q}}$. These states have the property that
$\braket{\mu_i}{\nu_j}=\frac{q}{2(q-1)}(1-\delta_{i,j}).$


\section{Function Evaluation}\label{sec:func}

In this section, we prove 

\begin{theorem}\label{thm:spEvalNew}
For $X\subseteq[q]^n$, let $P$ be a (scaled and normalized) span program that decides $f:X \rightarrow
\left\{0,1\right\}$ with witness size $W=\max_{\chi\in X}\w{P}{\chi}$.  Then there is a quantum algorithm that for any $x\in X$ evaluates $f(x)$ with probability $1-O(\delta)$ and uses $\tilde{O}\left(\sqrt{\w{P}{x}W}\log(1/\delta)\right)$ queries to the oracle $O_x$.
\end{theorem}
\noindent The algorithm might use more than
the stated number of queries or output the incorrect value, but the
probability of either or both of these events occurring is  $O(\delta).$

Comparing \cref{thm:spEvalNew} to \cref{cor:spEval}, we see that in the worst case, when we have an input $x$ where 
$\w{P}{x}=\max_{\chi\in X}\w{P}{\chi}$, the performance of our algorithm is the same, up to log factors,
as the standard span program algorithm. However, when we have an instance $x$ with a smaller value of $\w{P}{x}$,
then our algorithm has improved query complexity, without having to know anything ahead of time about the size of $\w{P}{x}$.

Our algorithm makes use of a subroutine that is similar to the standard
span program algorithm, but which will almost never
output $1$ unless the function has value $1$, while it might output $0$ even if
the value of the function is $1$. In other words, the subroutine has a negligible
probability of a false positive, but a potentially large probability of a
false negative.

We repeatedly run this subroutine while exponentially increasing
a parameter $\alpha$. As $\alpha$ gets larger, the probability of a false
negative decreases, while the runtime of the algorithm increases. If we see a negative
outcome at an intermediate round, since we do not know if this value is a true negative or false
negative, we continue running. We stop when we get a positive outcome, or when
$\alpha$ reaches its maximum possible value, at which point the probability of getting a false negative is also negligible.

Given a span program $P=(H,V,\tau,A)$ on $[q]^n$, let $\tilde{H}=H\oplus \textrm{span}\{\ket{\hat{0}}\},$ and $\tilde{H}(x)=H(x)\oplus \textrm{span}\{\ket{\hat{0}}\},$
where $\ket{\hat{0}}$ is orthogonal to $H$ and $V$. Then
we define the linear operator $\Aa{\alpha} \in \sop L(\tilde{H},V)$ as
\begin{equation}\label{eq:alphaIntro}
\Aa{\alpha}=\frac{1}{\alpha}\ketbra{\tau}{\hat{0}}+A.
\end{equation}
Let $\Lal{\alpha} \in \sop L(\tilde{H},\tilde{H})$ be the orthogonal
projection onto the kernel of $\Aa{\alpha}$, and let $\PtHx\in \sop
L(\tilde{H},\tilde{H})$ be the projection onto $\tilde{H}(x).$
Finally, let $\U{P}{x}{\alpha}=(2\PtHx-I)(2\Lal{\alpha}-I)$. Note that
$2\PtHx-I$ can be implemented with two applications of $O_x$ \cite[Lemma
3.1]{itoApproximateSpanPrograms2019}, and $2\Lal{\alpha}-I$ can be implemented
without any applications of $O_x$.


\begin{algorithm}
    \DontPrintSemicolon
    \SetKwInOut{Input}{Input}
    \SetKwInOut{Output}{Output}
    \Input{Error tolerance $\delta$, span program $P$ that decides a function $f$, oracle $O_x$}
    \Output{$f(x)$ with probability $1-O(\delta)$}
    $N\gets 9\log\left(\left\lceil\log\sqrt{3W}\right\rceil/\delta\right)/2;\quad \epsilon\gets 1/9$\;
    \For{$i\geq 0$ \KwTo $\left\lceil\log\sqrt{3W}\right\rceil$ }{
     $\alpha\gets2^i$\;
     Repeat $N$ times: Phase Checking of $\U{P}{x}{\alpha}$ on $\ket{\hat{0}}_A\ket{0}_B$ with error $\epsilon$, precision $\sqrt{\frac{\epsilon}{\alpha^2 W}}$\;
     \lIf{\em{Measure $\ket{0}$ in register $B$ at least $N/2$ times}}{
      \Return 1
     }
     Repeat $N$ times: Phase Checking of $\U{P^\dagger}{x}{\alpha}$ on $\ket{\hat{0}}_A\ket{0}_B$ with error $\epsilon$, precision $\sqrt{\frac{\epsilon}{\alpha^2 W}}$\;
     \lIf{\em{Measure $\ket{0}$ in register $B$ at least $N/2$ times}}{
      \Return 0
     }
    }
    \Return 1 \tcp{with low probability the algorithm makes a random guess}
    \caption{}
    \label{alg:bool}
\end{algorithm}


We will show \cref{alg:bool} solves the function decision problem correctly and with the desired complexity. We use the following lemma, which we prove in \cref{app:proofs}, to analyze the calls to  Phase Checking (see \cref{lem:phase_det})
in  Lines 4 and 6 of \cref{alg:bool}, which then allows us to prove \cref{thm:spEvalNew}.

\begin{restatable}{lemma}{phaseEstEarly}\label{lem:phase_est_early}
Let the span program $P$ decide the function $f$ with witness
size $W$, and let $C\geq 2$. Then for Phase Checking with unitary $U(P, \alpha)$ on the
state $\ket{\hat{0}}_A\ket{0}_B$ with error $\epsilon$ and precision $\Theta=\sqrt{\frac{\epsilon}{\alpha^2 W}}$,
\begin{enumerate}
\item If $f(x)=1$, and $\alpha^2\geq Cw(P,x)$, for a large enough constant $C$, then the probability of measuring the $B$ register to be in the state $\ket{0}_B$ is at least $1-1/C$.
\item If $f(x)=0$, the probability of measuring the $B$ register in the state $\ket{0}_B$ is at most $3\epsilon.$
\end{enumerate}
\end{restatable}


Note that if $f(x)=1$ and $C\w{P}{x} > \alpha^2$, \cref{lem:phase_est_early} makes no claims about
the output. However, since our algorithm can handle false negatives, as
discussed previously, this is acceptable. Without more information about the
phase gap of the unitary we are running Phase Checking on, it seems difficult
to obtain information about this regime.

To prove \cref{lem:phase_est_early}, we use an analysis that mirrors the
Boolean function decision algorithm of Belovs and Reichardt \cite[Section
5.2]{belovsSpanProgramsQuantum2012} and Cade et al. \cite[Section
C.2]{cadeTimeSpaceEfficient2018} and the dual adversary algorithm of Reichardt
\cite[Algorithm 1]{reichardtReflectionsQuantumQuery2011}. Our approach differs
from these previous algorithms in the addition of a parameter that controls
the precision of our phase estimation; we do not always run phase estimation
with a precision that is as high as those in previous works, which is what
causes our false negatives. This general algorithmic approach has
not (to the best of our knowledge\footnote{Jeffery and Ito
\cite{itoApproximateSpanPrograms2019}  also design a function decision
algorithm for non-Boolean span programs, but it has a few differences from our
approach and from that of Refs.
\cite{belovsSpanProgramsQuantum2012,cadeTimeSpaceEfficient2018}; for example,
the initial state of Jeffery and Ito's algorithm might require significant
time to prepare,  while our initial state can be prepared in $O(1)$ time.})
been applied to the the non-Boolean span program of \cref{def:SP}, so while
not surprising that it works in this setting, our analysis in \cref{app:proofs} may be
of independent interest for other applications.

\begin{proof}[Proof of \cref{thm:spEvalNew}]
We analyze \cref{alg:bool}.

Consider the case that $f(x)=1$. Let $T=\left\lceil\log\sqrt{3
\w{P}{x}}\right\rceil.$ We analyze the probability that the algorithm outputs
$1$ within the first $T$ iterations of the \texttt{for} loop. This is the
probability that the algorithm doesn't output $0$ in any of the first $T-1$
rounds (not outputting a $0$ includes the events of outputting a $1$ or
continuing to the next round), times the probability it outputs at $1$ in
the $T$th round. Note that because $\w{P}{x}\leq W$, and the \texttt{for} loop
can repeat up to $\left\lceil\log\sqrt{3 W}\right\rceil$ times, it is possible
to have $T$ iterations.

At the $T$th round,
\begin{equation}
 \alpha=2^{\left\lceil\log\sqrt{3
\w{P}{x}}\right\rceil}\geq \sqrt{3 \w{P}{x}}.
 \end{equation} 
 Since $\alpha^2\geq 3
\w{P}{x}$, by \cref{lem:phase_est_early}, the probability of measuring  $\ket{0}$
at a single repetition of Phase Checking is at least $2/3,$ so the probability 
of seeing at least $N/2$ results with outcome $\ket{0}$ and outputting $1$ is at least
\begin{equation}
 1-2^{-2\log e N/9}>1-2^{-2N/9}
 \end{equation} using Hoeffding's inequality \cite{Hoeff63} for the binomial distribution.

The probability of not outputting a $0$ at any previous round depends on Phase
Checking with $\U{P^\dagger}{x}{\alpha}$. From \cref{lem:SP_duals},
$P^\dagger$ decides the function $\neg f$, and since $f(x)=1$, we have $\neg
f(x)=0$. By \cref{lem:phase_est_early}, each time we run Phase Checking of
$\U{P^\dagger}{x}{\alpha}$, there is at most a $3\epsilon=1/3$ probability of
measuring $\ket{0}$. Thus the probability of seeing at least $N/2$ results with outcome $\ket{0}$ and outputting $0$ is at most $2^{-2N/9}$. Thus, even if we assume a worst case situation where we never
output a $1$ in the first $T-1$ rounds, the probability that we do not output
a $0$ over the first $T-1$ rounds is at least
\begin{equation}
\left(1-2^{-2N/9}\right)^{T-1}>1-T2^{-2N/9}.
\end{equation}
Thus our total success probability is at least
\begin{equation}
\left(1-T2^{-2N/9})\right)\left(1-2^{-2N/9}\right)=1-O(T2^{-2N/9})=1-O(\delta),
\end{equation}
where we've used that $T\leq \left\lceil\log\sqrt{3W}\right\rceil$ and $N=9\log\left(\left\lceil\log\sqrt{3W}\right\rceil/\delta\right)/2$.

For the case that $f(x)=0$, we have $\neg f(x)=1$, so since $P^\dagger$
decides $\neg f$ (by \cref{lem:SP_duals}), we can use the same analysis as in
the case of $f(x)=1$, with $\w{P}{x}$ replaced by $\w{P^\dagger}{x}$. By
\cref{lem:SP_duals}, we have $\w{P^\dagger}{x}\leq \w{P}{x}$, so the same analysis of the success probability will apply.

Now to analyze the query complexity. We've argued that with probability
$1-O(\delta)$, the algorithm terminates within $T$ rounds. In the $i$th round, by
\cref{lem:phase_det}, the number of queries required to run a single repetition of Phase Checking is
$O\left(2^{i}\sqrt{W}\right).$ Over the $N$ repetitions, we have that the $i$th round uses $O\left(2^{i}\sqrt{W}\log(\log (W)/\delta)\right)$ queries.
Summing over the rounds from $i=0$ to
$T=\left\lceil\log\sqrt{3\w{P}{x}}\right\rceil$ and using the geometric series formula,
we find that the total number of queries is
\begin{equation}
O\left(\sqrt{\w{P}{x}W}\log(\log(W)/\delta)\right)=\tilde{O}\left(\sqrt{\w{P}{x} W}\log(1/\delta)\right).
\end{equation}
\end{proof}

\subsection{Application to st-connectivity}

As an example, we apply \cref{thm:spEvalNew} to the problem of
$st$-connectivity on an $n$ vertex graph. There is a span program
$P$ such that (after scaling), for inputs $x$ where there is a path from $s$
to $t$, $w(P,x)=R_{s,t}(x)\sqrt{n}$ where $R_{s,t}(x)$ is the effective
resistance from $s$ to $t$ on the subgraph induced by $x$,
 and for inputs $x$ where there is not a
path from $s$ to $t$, $w(P,x)=C_{s,t}(x)/\sqrt{n}$, where $C_{s,t}(x)$ is the
effective capacitance between $s$ and $t$
\cite{belovsSpanProgramsQuantum2012,jarretQuantumAlgorithmsConnectivity2018}. In an $n$-vertex graph, the effective resistance is at
most $n$, and the effective capacitance is at most $n^{2}$. Thus if we desire a bounded error algorithm, by \cref{thm:spEvalNew}, we can determine whether or not there is a path with $\tilde{O}(\sqrt{R_{s,t}(x)n^2})$ queries if there there is a path, and
$\tilde{O}(\sqrt{C_{s,t}(x)n})$ queries if there is not a path. The 
effective resistance is at most the shortest path between two vertices,
and the effective capacitance is at most the smallest cut between two
vertices. Thus
our algorithm determines whether or not there is a path from $s$ to
$t$ with $\tilde{O}(\sqrt{k}n)$ queries if there is a path of length $k$, and if there is no path,
the algorithm uses $\tilde{O}(\sqrt{nc})$ queries, where $c$ is the size of
the smallest cut between $s$ and $t$. Importantly, one does not need to know
bounds on $k$ or $c$ ahead of time to achieve this query complexity.

The analysis of the other examples listed in \cref{sec:intro} is similar.


\section{State Conversion}\label{sec:StateConv}

In this section, we prove the following result regarding quantum state conversion:
\begin{theorem}\label{thm:stateConvNew}
Let $\sop{P}$ be a (normalized) converting
vector set from $\rho$ to $\sigma$ with witness size $W=\min\left\{\max_{\chi\in X}w_+(\sop P,\chi),\max_{\chi\in X}w_-(\sop P,\chi)\}\right\}$. Then there is quantum algorithm that for any $x\in X$ with probability $1-p$ converts $\ket{\rho_x}$ to $\ket{\sigma_x}$ with error $\varepsilon$ and uses 
\\$\tilde{O}\left(\sqrt{\min\{w_+(\sop P,x),w_-(\sop P,x)\}W}/\left(\varepsilon^5p\right)\right)$ queries to the oracle $O_x$.
\end{theorem}

Comparing \cref{thm:stateConvNew} with \cref{cor:stateConv}, and for a moment ignoring the scaling with $\epsilon$ and $p$, 
 we see that in the worst case, when we have an input $x$ where 
$\min\{w_+(\sop P,x),w_-(\sop P,x)\}=
\min\left\{\max_{\chi\in X}w_+(\sop P,\chi),\max_{\chi\in X}w_-(\sop P,\chi)\}\right\}$
 the performance of our algorithm is the same, up to log factors,
as the standard state conversion algorithm. However, when we have an instance $x$ with a smaller value of $w_\pm(\sop P,x)$,
then our algorithm has improved query complexity, without having to know anything about the witness size of our input ahead of time.

Our
algorithm has worse scaling in $\varepsilon$ than \cref{cor:stateConv}, but when $\varepsilon$ is a constant this is a non-issue. We also note that
while the $\tilde{O}$ in \cref{cor:stateConv} hides logarithmic factors in
$\varepsilon$, in \cref{thm:stateConvNew} it hides logarithmic factors in both $\varepsilon$ and
$W$. We believe that it should be possible to improve the the $1/p$ factor in the complexity to $\log(1/p)$.

The problem of state conversion is a more general problem than function evaluation, and it can be used to solve the function evaluation problem. However, because of the worse scaling with $\varepsilon$ in \cref{thm:stateConvNew}, we considered function evaluation separately (see \cref{sec:func}).

There might be some inputs $x\in X$ for which the problem of
converting from $\ket{\rho_x}$ to $\ket{\sigma_x}$ is less difficult, and thus requires fewer queries.
However if we run the standard state conversion algorithm for less than the
worst-case queries, we can not tell whether the computation has
completed, since the output is
a state $\ket{\sigma_x}$ where $x$ is unknown, and any measurement will collapse the state. This contrasts with function evaluation, where at least there was a measurement at the end of the computation.

Thus instead of repeatedly running the standard state conversion algorithm for
increasingly longer times, as with function decision, we instead use an initial probing
protocol that we repeatedly run for increasingly longer times. This probing
subroutine helps us determine how long we need
to run the main state conversion algorithm for to guarantee success.

In the following we use most of the notation conventions of Ref.
\cite{leeQuantumQueryComplexity2011} for clarity. For $X\subseteq [q]^n,$ let
$\sop{P}=\left(\{\ket{u_{xj}}\},\{\ket{v_{yj}}\}\right)$ be a converting
vector set from $\rho$ to $\sigma$, where for all $x\in X$, the states $\ket{\rho_x}$ and $\ket{\sigma_x}$ are in the Hilbert space $\sop H$. For all $x\in X$, define 
$\ket{t_{x\pm}},
\ket{\psi_{x,\alpha,\hat{\varepsilon}}}\in(\mathbb{C}^2\otimes \sop H)\oplus(\mathbb{C}^n\otimes \mathbb{C}^q\otimes \mathbb{C}^m)$ as
\begin{align}
\ket{t_{x\pm}}=\frac{1}{\sqrt{2}}\left(\ket{0}\ket{\rho_x}\pm\ket{1}\ket{\sigma_x}\right), \quad \textrm{ and }\quad
\ket{\psi_{x,\alpha,\hat{\varepsilon}}}=\sqrt{\frac{\hat{\varepsilon}}{\alpha}}\ket{t_{x-}}-\sum_{j\in[n]}\ket{j}\ket{\mu_{x_j}}\ket{u_{xj}},
\end{align}
where $\ket{\mu_{x_j}}$ is from \cref{eq:mu_nu}, and $\alpha$ is a parameter
 analogous to the parameter $\alpha$ in \cref{eq:alphaIntro}. We will choose 
$\hat{\varepsilon}$ to achieve a desired
accuracy of $\varepsilon$ in our state conversion procedure. 
Set $\Lalp$ to equal the projection onto the orthogonal complement of the span
of the vectors $\{\ket{\psi_{x,\alpha,\hat{\varepsilon}}}\}_{x\in X}$, and set
$\PHxp=I-\sum_{j\in[n]}\proj{j}\otimes \proj{\mu_{x_j}}\otimes
I_{\mathbb{C}^n}$. Finally, we set $\Up{\sop
P}{x}{\alpha}{\hat{\varepsilon}}=(2\PHxp-I)(2\Lalp-I)$. The reflection $2\PHxp-I$ can be
implemented with two applications of $O_x$
\cite{leeQuantumQueryComplexity2011}, and the reflection $(2\Lalp-I)$ is independent of $x$ and so requires no queries.


\begin{algorithm}[h!]
    \DontPrintSemicolon
    \SetKwInOut{Input}{Input}
    \SetKwInOut{Output}{Output}
    \Input{Converting vector set $\sop P$ from $\rho$ to $\sigma$ with witness size $W$, failure probability $p$, error $\varepsilon$, oracle $O_x$, initial state $\ket{\rho_x}$}
    \Output{$\ket{\tilde{\sigma}_x}$ such that $\|\ket{\tilde{\sigma}_x}-\ket{1}\ket{\sigma_x}\ket{0}\|\leq \varepsilon$}
    \tcc{Probing Stage}
    $\hat{\varepsilon}\gets \varepsilon^2/9$\; \label{algline:one}
    \For{i=0 \KwTo $\lceil\log W\rceil$}
    {
      $\alpha\gets2^i$\;
      \For{$\sop P'\in\{\sop P, \sop P^C\}$}{
      $\mathcal{A}\gets D(\Up{\sop P'}{x}{\alpha}{\hat{\varepsilon}})$ (\cref{lem:phase_det}) to precision $\hat{\varepsilon}^{3/2}/\sqrt{\alpha W}$ and accuracy $\hat{\varepsilon}^2$ \;              
        $\hat{a}\gets$ Amplitude Estimation (\cref{lem:ampEst}) of probability of outcome $\ket{0}_B$ in register $B$ when  $\mathcal{A}$ acts on $(\ket{0}\ket{\rho_x})_A\ket{0}_B$ to additive error $\hat{\varepsilon}/4$ with probability of failure $\frac{p}{\log W}$ \;
        \lIf{$\hat{a}-1/2>-\frac{11}{4}\hat{\varepsilon}$}
        {
          Continue to State Conversion Stage
        }
      }

    }
    \tcc{State Conversion Stage}

      Apply $R(\Up{\sop P'}{x}{\alpha}{\hat{\varepsilon}})$ (\cref{lem:phase_refl}) with precision $\hat{\varepsilon}^{3/2}/\sqrt{\alpha W}$ and accuracy $\hat{\varepsilon}^2$ to $(\ket{0}\ket{\rho_x})_A\ket{0}_B$  and output the result\;
    
    \caption{}
    \label{alg:StateConv}
\end{algorithm}{}



The idea of our approach, given in \cref{alg:StateConv}, is that when we apply Phase Reflection of $\Up{\sop
 P'}{x}{\alpha}{\hat{\varepsilon}}$ in Line 8 to
 $(\ket{0}\ket{\rho_x})_A\ket{0}_B=\frac{1}{\sqrt{2}}(\ket{t_{x+}}_A\ket{0}_B+\ket{t_{x-}}_A\ket{0}_B)$,
 we want $\ket{t_{x+}}_A\ket{0}_B$ to pick up a $+1$ phase, and
 $\ket{t_{x-}}_A\ket{0}_B$ to pick up a $-1$ phase. If this happened perfectly,
  we would have the desired state $(\ket{1}\ket{\sigma_x})_A\ket{0}_B$.
 In \cref{claim:stopping,claim:ampEstbound,claim:spec_gap_applied,lem:finalStateAnalysis} we derive results that show that in the State Conversion stage of \cref{alg:StateConv},
 $\ket{t_{x+}}_A\ket{0}_B$ will mostly pick up a $+1$ phase, and
 $\ket{t_{x-}}_A\ket{0}_B$ will mostly pick up a $-1$ phase, resulting in a state
 close to $(\ket{1}\ket{\sigma_x})_A\ket{0}_B$.

In
 \cref{claim:spec_gap_applied,claim:ampEstbound,claim:stopping,lem:finalStateAnalysis}, let
 $\sop{P}=\left(\{\ket{u_{xj}}\},\{\ket{v_{yj}}\}\right)$ be a converting
 vector set from $\rho$ to $\sigma$ with $W=W(\sop P)$. When we write $\Pi_0(\Up{\sop
 P}{x}{\alpha}{\hat{\varepsilon}})$, $\overline{\Pi}_0(\Up{\sop
 P}{x}{\alpha}{\hat{\varepsilon}})$ or $R(\Up{\sop
 P}{x}{\alpha}{\hat{\varepsilon}})$, it refers to Phase 
 Checking/Reflection on $\Up{\sop
 P}{x}{\alpha}{\hat{\varepsilon}}$ with precision $\hat{\varepsilon}^{3/2}/\sqrt{\alpha W}$
 and accuracy $\hat{\varepsilon}^2$. 

\begin{restatable}{lemma}{specBound}\label{claim:spec_gap_applied}
If $\Theta = \hat{\varepsilon}^{3/2}/\sqrt{\alpha W}$, then $\| P_\Theta(\Up{\sop
 P}{x}{\alpha}{\hat{\varepsilon}})\ket{t_{x-}} \|^2 \leq \frac{\hat{\varepsilon}^2}{2}$.
\end{restatable}

\vspace{-.2cm}

\begin{restatable}{lemma}{half}\label{claim:stopping}
If $\alpha\geq w_+(\sop P,x)$, then
$\|\Pi_0(\Up{\sop
 P}{x}{\alpha}{\hat{\varepsilon}})({\ket{0}\ket{\rho_x}})_A\ket{0}_B\|^2\geq \frac{1}{2}\left(1-5\hat{\varepsilon}\right)$.
\end{restatable}
\vspace{-.2cm}
\begin{restatable}{lemma}{wrongPhase}\label{claim:ampEstbound}
If
 $\|\Pi_0(\Up{\sop
 P}{x}{\alpha}{\hat{\varepsilon}})\ket{0}\ket{\rho_x}\ket{0}\|^2\geq1/2-3\hat{\varepsilon}$, then
 $\|\overline{\Pi}_0(\Up{\sop
 P}{x}{\alpha}{\hat{\varepsilon}})\ket{t_{x+}}\ket{0}\|^2\leq10\hat{\varepsilon}$.
\end{restatable}
\vspace{-.3cm}
\begin{restatable}{lemma}{finalState}\label{lem:finalStateAnalysis}
If  $\|\Pi_0(\Up{\sop
 P}{x}{\alpha}{\hat{\varepsilon}})(\ket{0}\ket{\rho_x})_A\ket{0}_B\|\geq 1/2-3\hat{\varepsilon}$, then \\
$\|R(\Up{\sop
 P}{x}{\alpha}{\hat{\varepsilon}})(\ket{0}\ket{\rho_x})_A\ket{0}_B-(\ket{1}\ket{\sigma_x})_A\ket{0}_B\|\leq 6\sqrt{\hat{\varepsilon}}$
\end{restatable}

\cref{claim:spec_gap_applied} ensures that the $\ket{t_{x-}}$ portion of the state will mostly pick up a $-1$ phase. The proof closely follows \cite[Claim 4.5]{leeQuantumQueryComplexity2011}. \cref{claim:stopping} ensures the Probing Routine halts when $\alpha$ reaches the witness size. The proof follows \cite[Claim 4.4]{leeQuantumQueryComplexity2011} to show that $\|P_0\ket{t_{x+}}\|^2$ is close to $1$ (which also ensures that the $\ket{t_{x+}}$ portion of the state will mostly pick up a $1$ phase), and then with a little bit of work and \cref{claim:spec_gap_applied}, we can prove that $\|\Pi_0(\Up{\sop
 P}{x}{\alpha}{\hat{\varepsilon}})({\ket{0}\ket{\rho_x}})_A\ket{0}_B\|^2$ is close to $1/2.$ 
With \cref{claim:ampEstbound}, we bound the amount of $\ket{t_{x+}}$ that can pick up
an incorrect phase of $-1$. The proof uses the triangle inequality and
our previous lemmas.
\cref{lem:finalStateAnalysis} uses \cref{claim:ampEstbound} and \cref{claim:spec_gap_applied} to show that after a successful Probing Stage, the State Conversion stage of \cref{alg:StateConv} produces a state
close to our target. It is our version of \cite[Proposition
4.6]{leeQuantumQueryComplexity2011}, proven using slightly weaker bounds. See \cref{app:proofs} for proofs of these Lemmas.

We now use \cref{claim:stopping}  and \cref{lem:finalStateAnalysis} to prove \cref{thm:stateConvNew}:

\begin{proof}[Proof of \cref{thm:stateConvNew}]
We analyze \cref{alg:StateConv}. With probability $p$, the Probing Stage will stop with assignments 
of $\alpha$ and $\sop P'$ such that $\|\Pi_0(\Up{\sop
 P'}{x}{\alpha}{\hat{\varepsilon}}){\ket{0}\ket{\rho_x}}_A\ket{0}_B\|^2>\frac{1}{2}-3\hat{\varepsilon}.$
This is because during the Probing Stage, we increase $\alpha$ until it is larger than $W$, so we
eventually have $\alpha\geq\min\{w_+(\sop P,x),w_+(\sop P^C,x)\}$, which by   \cref{claim:stopping}, implies $\|\Pi_0(\Up{\sop
 P'}{x}{\alpha}{\hat{\varepsilon}}){\ket{0}\ket{\rho_x}}_A\ket{0}_B\|^2>\frac{1}{2}-5/2\hat{\varepsilon}$
for $\sop P'=\sop P$ if $\alpha\geq w_+(\sop P,x)$, and $\sop P'=\sop P^C$ if
$\alpha\geq w_-(\sop P,x)$. This ends the Probing Stage since
 amplitude amplification estimates the value of $\|\Pi_0(\Up{\sop
P'}{x}{\alpha}{\hat{\varepsilon}}){\ket{0}\ket{\rho_x}}_A\ket{0}_B\|^2$ to within $\frac{1}{4}\hat{\varepsilon}$. Finally, as there are $O(\log W)$ rounds, and each fails with probability $O(p/\log W)$, the probability that all rounds of amplitude
amplification are successful is $O(p)$.

The Probing Stage may stop before $\alpha\geq\min\{w_+(\sop P,x),w_+(\sop P^C,x)\}$, but (assuming no errors in this stage), it is guaranteed to stop when $\|\Pi_0(\Up{\sop
 P'}{x}{\alpha}{\hat{\varepsilon}}){\ket{0}\ket{\rho_x}}_A\ket{0}_B\|^2>\frac{1}{2}-3\hat{\varepsilon}$, 
since the value of $\|\Pi_0(\Up{\sop
 P'}{x}{\alpha}{\hat{\varepsilon}}){\ket{0}\ket{\rho_x}}_A\ket{0}_B\|^2$ will be within $\frac{1}{4}\hat{\varepsilon}$ of its estimated value by our parameter choices
for amplitude estimation.

Applying  \cref{lem:finalStateAnalysis}, since $\|\Pi_0(U){\ket{0}\ket{\rho_x}}_A\ket{0}_B\|^2\geq
1/2-3\hat{\varepsilon}$, the State Conversion stage 
produces a state $\ket{\tilde{\sigma}}=R(\Up{\sop
 P'}{x}{\alpha}{\hat{\varepsilon}})(\ket{0}\ket{\rho_x})_A\ket{0}_B$ such that
 \begin{equation}
 \|\ket{\tilde{\sigma}}-(\ket{1}\ket{\sigma_x})_A\ket{0}_B\|\leq 3\sqrt{\hat{\varepsilon}}=\varepsilon.
 \end{equation}
Thus the algorithm is correct with the stated success probability and accuracy.

Now we analyze the query complexity. From our previous discussion, the Probing Stage stops by the
round when $\alpha$ is at least $\min\{w_+(\sop
P,x),w_+(\sop P^C,x)\}$. At the $i$th round of the Probing Stage,
phase estimation with precision $\hat{\varepsilon}^{3/2}/\sqrt{2^i W}$ and
accuracy $\hat{\varepsilon}^2$, which $O\left(\frac{\sqrt{2^i
W}}{\hat{\varepsilon}^{3/2}}\log(\frac{1}{\hat{\varepsilon}})\right)$ queries is implemented implement $O\left(\frac{\log
W}{\hat{\varepsilon}p}\right)$ times inside the amplitude estimation
subroutine. Thus the queries used by the Probing Stage of the algorithm is
\begin{equation*}
\sum_{i=0}^{\log \lceil \min\{w_+(\sop
P,x),w_+(\sop P^C,x)\}\rceil}O\left(\frac{\sqrt{2^i W}}{\hat{\varepsilon}^{3/2}}\log\left(\frac{1}{\hat{\varepsilon}}\right)\frac{\log W}{\hat{\varepsilon}p}\right)
=\tilde{O}\left(\frac{\sqrt{\min\{w_+(\sop
P,x),w_+(\sop P^C,x)\} W}}{\hat{\varepsilon}^{5/2}p}\right).
\end{equation*}

Finally, the Phase Reflection stage of the algorithm has precision
$\hat{\varepsilon}^{3/2}/\sqrt{2^i W}$ and
accuracy $\hat{\varepsilon}^2$, so uses another $O\left(\sqrt{\min\{w_+(\sop
P,x),w_+(\sop P^C,x)\}
W}\log(\frac{1}{\hat{\varepsilon}})/(\hat{\varepsilon}^{5/2})\right)$ queries. Thus the probing state dominates, and the total number of queries used by the algorithm is (using \cref{lem:complement}) $\tilde{O}\left(\sqrt{\min\{w_+(\sop
P,x),w_-(\sop P,x)\}W}/(\varepsilon^5p)\right),$
 where we've used that $\hat{\varepsilon}=\varepsilon^2/9$ (Line 2 of \cref{alg:StateConv}).
\end{proof}

\section{Acknowledgments}

We thank Stacey Jeffery for valuable discussions, especially for her preliminary notes on span program negation, and the anonymous ICALP referees for insightful suggestions.

This research was sponsored by the Army Research Office and was accomplished under Grant Number W911NF-20-1-0327. The views and conclusions contained in this document are those of the authors and should not be interpreted as representing the official policies, either expressed or implied, of the Army Research Office or the U.S. Government. The U.S. Government is authorized to reproduce and distribute reprints for Government purposes notwithstanding any copyright notation herein.

\bibliography{spanPrograms.bib}

\begin{thebibliography}{26}
\providecommand{\natexlab}[1]{#1}
\providecommand{\url}[1]{\texttt{#1}}
\expandafter\ifx\csname urlstyle\endcsname\relax
  \providecommand{\doi}[1]{doi: #1}\else
  \providecommand{\doi}{doi: \begingroup \urlstyle{rm}\Url}\fi

\bibitem[Aharonov(1999)]{aharonovQuantumComputation1999}
Dorit Aharonov.
\newblock Quantum {{Computation}}.
\newblock \emph{Annual Reviews of Computational Physics VI}, pages 259--346,
  1999.
\newblock \doi{10.1142/9789812815569_0007}.

\bibitem[Ambainis and De~Wolf(2001)]{ambainis2001average}
Andris Ambainis and Ronald De~Wolf.
\newblock Average-case quantum query complexity.
\newblock \emph{Journal of Physics A: Mathematical and General}, 34\penalty0
  (35):\penalty0 6741, 2001.

\bibitem[Beigi and Taghavi(2019)]{beigiQuantumSpeedupBased2019}
Salman Beigi and Leila Taghavi.
\newblock Quantum {{Speedup Based}} on {{Classical Decision Trees}}.
\newblock \emph{arXiv:1905.13095}, 2019.

\bibitem[Belovs and Reichardt(2012)]{belovsSpanProgramsQuantum2012}
A.~Belovs and B.W. Reichardt.
\newblock Span programs and quantum algorithms for st-connectivity and claw
  detection.
\newblock \emph{Lecture Notes in Computer Science}, 7501 LNCS:\penalty0
  193--204, 2012.
\newblock \doi{10.1007/978-3-642-33090-2_18}.

\bibitem[Belovs(2012)]{belovsSpanProgramsFunctions2012}
Aleksandrs Belovs.
\newblock Span programs for functions with constant-sized 1-certificates:
  Extended abstract.
\newblock In \emph{Proceedings of the Forty-Fourth Annual {{ACM}} Symposium on
  {{Theory}} of Computing}, {{STOC}} '12, pages 77--84, 2012.
\newblock \doi{10.1145/2213977.2213985}.

\bibitem[Belovs and Rosmanis(2020)]{belovsTightQuantumLower2020a}
Aleksandrs Belovs and Ansis Rosmanis.
\newblock Tight {{Quantum Lower Bound}} for {{Approximate Counting}} with
  {{Quantum States}}.
\newblock \emph{arXiv:2002.06879}, 2020.

\bibitem[Boyer et~al.(1998)Boyer, Brassard, H{\o}yer, and
  Tapp]{boyerTightBoundsQuantum1998}
Michel Boyer, Gilles Brassard, Peter H{\o}yer, and Alain Tapp.
\newblock Tight {{Bounds}} on {{Quantum Searching}}.
\newblock \emph{Fortschritte der Physik}, 46\penalty0 (4-5):\penalty0 493--505,
  1998.
\newblock
  \doi{10.1002/(SICI)1521-3978(199806)46:4/5<493::AID-PROP493>3.0.CO;2-P}.

\bibitem[Brassard et~al.(1998)Brassard, H{\o}yer, and
  Tapp]{brassardQuantumCounting1998}
Gilles Brassard, Peter H{\o}yer, and Alain Tapp.
\newblock Quantum counting.
\newblock In \emph{Automata, {{Languages}} and {{Programming}}}, pages
  820--831, 1998.
\newblock \doi{10.1007/BFb0055105}.

\bibitem[Brassard et~al.(2000)Brassard, H{\o}yer, Mosca, and
  Tapp]{brassardQuantumAmplitudeAmplification2000}
Gilles Brassard, Peter H{\o}yer, Michele Mosca, and Alain Tapp.
\newblock Quantum {{Amplitude Amplification}} and {{Estimation}}.
\newblock \emph{arXiv:quant-ph/0005055}, 2000.

\bibitem[Cade et~al.(2018)Cade, Montanaro, and
  Belovs]{cadeTimeSpaceEfficient2018}
Chris Cade, Ashley Montanaro, and Aleksandrs Belovs.
\newblock Time and space efficient quantum algorithms for detecting cycles and
  testing bipartiteness.
\newblock \emph{Quantum Information \& Computation}, 18\penalty0
  (1-2):\penalty0 18--50, 2018.

\bibitem[Cleve et~al.(1998)Cleve, Ekert, Macchiavello, and
  Mosca]{cleveQuantumAlgorithmsRevisited1998}
R.~Cleve, A.~Ekert, C.~Macchiavello, and M.~Mosca.
\newblock Quantum algorithms revisited.
\newblock \emph{Proceedings of the Royal Society of London. Series A:
  Mathematical, Physical and Engineering Sciences}, 454\penalty0
  (1969):\penalty0 339--354, 1998.
\newblock \doi{10.1098/rspa.1998.0164}.

\bibitem[DeLorenzo et~al.(2019)DeLorenzo, Kimmel, and
  Witter]{delorenzoApplicationsQuantumAlgorithm2019a}
Kai DeLorenzo, Shelby Kimmel, and R.~Teal Witter.
\newblock Applications of the {{Quantum Algorithm}} for st-{{Connectivity}}.
\newblock In \emph{14th {{Conference}} on the {{Theory}} of {{Quantum
  Computation}}, {{Communication}} and {{Cryptography}} ({{TQC}} 2019)}, volume
  135, pages 6:1--6:14, 2019.

\bibitem[Farhi et~al.(2007)Farhi, Goldstone, and Gutmann]{farhi2007quantum}
Edward Farhi, Jeffrey Goldstone, and Sam Gutmann.
\newblock A quantum algorithm for the hamiltonian nand tree.
\newblock \emph{arXiv preprint quant-ph/0702144}, 2007.

\bibitem[Grover(1997)]{groverQuantumMechanicsHelps1997a}
Lov~K. Grover.
\newblock Quantum {{Mechanics Helps}} in {{Searching}} for a {{Needle}} in a
  {{Haystack}}.
\newblock \emph{Physical Review Letters}, 79\penalty0 (2):\penalty0 325--328,
  1997.
\newblock \doi{10.1103/PhysRevLett.79.325}.

\bibitem[Hoeffding(1963)]{Hoeff63}
Wassily Hoeffding.
\newblock Probability inequalities for sums of bounded random variables.
\newblock \emph{Journal of the American Statistical Association}, 58\penalty0
  (301):\penalty0 13--30, 1963.
\newblock \doi{10.1080/01621459.1963.10500830}.
\newblock URL
  \url{http://www.tandfonline.com/doi/abs/10.1080/01621459.1963.10500830}.

\bibitem[Ito and Jeffery(2019)]{itoApproximateSpanPrograms2019}
Tsuyoshi Ito and Stacey Jeffery.
\newblock Approximate {{Span Programs}}.
\newblock \emph{Algorithmica}, 81\penalty0 (6):\penalty0 2158--2195, 2019.
\newblock \doi{10.1007/s00453-018-0527-1}.

\bibitem[Jarret et~al.(2018)Jarret, Jeffery, Kimmel, and
  Piedrafita]{jarretQuantumAlgorithmsConnectivity2018}
Michael Jarret, Stacey Jeffery, Shelby Kimmel, and Alvaro Piedrafita.
\newblock Quantum {{Algorithms}} for {{Connectivity}} and {{Related Problems}}.
\newblock In \emph{26th {{Annual European Symposium}} on {{Algorithms}}
  ({{ESA}} 2018)}, volume 112 of \emph{Leibniz {{International Proceedings}} in
  {{Informatics}} ({{LIPIcs}})}, pages 49:1--49:13, 2018.
\newblock \doi{10.4230/LIPIcs.ESA.2018.49}.

\bibitem[Kimmel(2011)]{kimmel2011quantum}
Shelby Kimmel.
\newblock Quantum adversary (upper) bound.
\newblock \emph{Chicago Journal of Theoretical Computer Science}, 2013\penalty0
  (4), 2011.

\bibitem[Kitaev(1995)]{kitaevQuantumMeasurementsAbelian1995}
A.~Yu Kitaev.
\newblock Quantum measurements and the {{Abelian Stabilizer Problem}}.
\newblock 1995.

\bibitem[Lee et~al.(2011)Lee, Mittal, Reichardt, {\v S}palek, and
  Szegedy]{leeQuantumQueryComplexity2011}
Troy Lee, Rajat Mittal, Ben~W. Reichardt, Robert {\v S}palek, and Mario
  Szegedy.
\newblock Quantum {{Query Complexity}} of {{State Conversion}}.
\newblock In \emph{2011 {{IEEE}} 52nd {{Annual Symposium}} on {{Foundations}}
  of {{Computer Science}}}, pages 344--353, 2011.
\newblock \doi{10.1109/FOCS.2011.75}.

\bibitem[Magniez et~al.(2011)Magniez, Nayak, Roland, and
  Santha]{magniezSearchQuantumWalk2011}
Fr{\'e}d{\'e}ric Magniez, Ashwin Nayak, J{\'e}r{\'e}mie Roland, and Miklos
  Santha.
\newblock Search via {{Quantum Walk}}.
\newblock \emph{SIAM Journal on Computing}, 40\penalty0 (1):\penalty0 142--164,
  2011.
\newblock \doi{10.1137/090745854}.

\bibitem[Montanaro(2010)]{montanaro2010quantum}
Ashley Montanaro.
\newblock Quantum search with advice.
\newblock In \emph{Conference on Quantum Computation, Communication, and
  Cryptography}, pages 77--93. Springer, 2010.

\bibitem[Reichardt(2009)]{reichardtSpanProgramsQuantum2009}
Ben~W. Reichardt.
\newblock Span programs and quantum query complexity: {{The}} general adversary
  bound is nearly tight for every boolean function.
\newblock \emph{50th Annual IEEE Symposium on Foundations of Computer Science},
  pages 544--551, 2009.
\newblock \doi{10.1109/FOCS.2009.55}.

\bibitem[Reichardt(2011)]{reichardtReflectionsQuantumQuery2011}
Ben~W. Reichardt.
\newblock Reflections for quantum query algorithms.
\newblock In \emph{Proceedings of the 2011 {{Annual ACM}}-{{SIAM Symposium}} on
  {{Discrete Algorithms}}}, Proceedings, pages 560--569. 2011.
\newblock \doi{10.1137/1.9781611973082.44}.

\bibitem[Yoder et~al.(2014)Yoder, Low, and
  Chuang]{yoderFixedPointQuantumSearch2014}
Theodore~J. Yoder, Guang~Hao Low, and Isaac~L. Chuang.
\newblock Fixed-{{Point Quantum Search}} with an {{Optimal Number}} of
  {{Queries}}.
\newblock \emph{Physical Review Letters}, 113\penalty0 (21):\penalty0 210501,
  2014.
\newblock \doi{10.1103/PhysRevLett.113.210501}.

\bibitem[Zhan et~al.(2012)Zhan, Kimmel, and Hassidim]{zhan2012super}
Bohua Zhan, Shelby Kimmel, and Avinatan Hassidim.
\newblock Super-polynomial quantum speed-ups for boolean evaluation trees with
  hidden structure.
\newblock In \emph{Proceedings of the 3rd Innovations in Theoretical Computer
  Science conference}, pages 249--265. ACM, 2012.

\end{thebibliography}
\bibliographystyle{plainnat}


\appendix
\section{Proofs from \cref{sec:Prelims}}\label{app:sec2}

\SPduals*

\begin{proof}
We first define $H'$, starting from $H'_{j,a}$:
\begin{equation}
H'_{j,a}=\textrm{span}\{\ket{v}: \ket{v}\in H_j \textrm{ and } \ket{v}\in H_{j,a}^\perp\},
\end{equation}
where $H_{j,a}^\perp$ is the orthogonal complement of $H_{j,a}.$ We define $H'_j=\sum_{a\in [q]}H'_{j,a}$, and $H'_{true}=H_{false}$ and $H'_{false}=H_{true}$.
Then
\begin{equation}
H'=H'_1\oplus H'_2\cdots H'_n\oplus H'_{true}\oplus H'_{false}.
\end{equation}
Let $\ket{\tilde{0}}$ be a vector that is orthogonal to $H$ and $V$, and
define $V'=H\oplus\textrm{span}\{\ket{\tilde{0}}\}$ and $\tau'=\ket{\tilde{0}}.$
Finally, set
\begin{equation}
A'=\ketbra{\tilde{0}}{w_0}+\Pi_H\Lambda_A,
\end{equation}
where $\Lambda_A$ is the projection onto the kernel of $A$, $\Pi_H$ is the projection onto $H$, and 
\begin{equation}
\ket{w_0}=\argmin_{\substack{\ket{v}\in H: A\ket{v}=\tau}}\|\ket{v}\|.
\end{equation}

Let $x\in X$ be an input with $f(x)=1$, so $x$ has a positive witness
$\ket{w}$ in $P$. We will show $\omega'=\bra{\tilde{0}}+\bra{w}$ is a negative
witness for $x$ in $P^\dagger$. Note $\omega'\tau'=1$, and also,
\begin{align}
\omega' A'&=(\bra{\tilde{0}}+\bra{w})(\ketbra{\tilde{0}}{w_0}+\Pi_H\Lambda_A)\\
&=\bra{w_0}+\bra{w}\Pi_H\Lambda_A\\
&=\bra{w_0}+\bra{w}\Lambda_A\\
&=\bra{w},
\end{align}
where in the second line, we have used that $\bra{\tilde{0}}\Pi_H\Lambda_A=0$
and $\braket{w}{\tilde{0}}=0$ because $\ket{\tilde{0}}$ is orthogonal to $H$.
The final line follows from \cite[Definition
2.12]{itoApproximateSpanPrograms2019}, which showed that every positive
witness can be written as $\ket{w}=\ket{w_0}+\ket{w^\perp}$, where
$\ket{w^\perp}$ is in the kernel of $A$ and $\ket{w_0}$ is orthogonal to the
kernel of $A$. 

Then $\bra{w}\Pi_{H'(x)}=0$, because $\ket{w}\in H(x)$, and
$H'(x)$ is orthogonal to $H(x),$ so $\omega'$ is a negative witness for $x$ in
$P^\dagger$.  Also, $\|\omega' A'\|^2=\|\ket{w}\|^2$, so the witness size of
this negative witness in $P^\dagger$ is the same as the corresponding
positive witness in $P.$

If $f(x)=0$, there is a negative witness $\omega$ for $x$ in $P$. Consider
$\ket{w'}=(\omega A)^\dagger$. Then
\begin{align}
A'\ket{w'}&=(\ketbra{\tilde{0}}{w_0}+\Pi_H\Lambda_A)(\omega A)^\dagger\\
&=\ket{\tilde{0}}(\omega A\ket{w_0})^\dagger\\
&=\ket{\tilde{0}}(\omega\tau)^\dagger\\
&=\ket{\tilde{0}},
\end{align}
where in the first line, we've used that $(\omega A)^\dagger$ is orthogonal to
the kernel of $A.$ Also, $\Pi_{H(x)}(\omega A)^\dagger=0$, so $\ket{w'}\in
H'(x)$. This means $\ket{w'}$ is a positive witness for $x$ in $P^\dagger.$
Also, $\|\ket{w'}\|^2=\|\omega A\|^2$, so the witness size of this positive
witness in $P^\dagger$ is the same as the corresponding negative
witness in $P.$

We have proved that $\w{P}{x}\geq \w{P^\dagger}{x}$ for all $x$, and that if
$x$ has a positive witness in $P$, it has a negative witness in $P^\dagger$,
and vice versa. This shows that $P^\dagger$ does indeed decide $\neg f.$
\end{proof}

\complment*
\begin{proof}
We will prove the more general result showing that this holds when the matrix $B$ in \cref{def:Gamma2} is
Hermitian, and the matrices $Z$ are symmetric (as is true when $B=\rho-\sigma$
and $Z=\Delta$ for a converting vector set). For all $x\in X$ and $j\in[n]$,
define
\begin{equation}
\ket{u^C_{xj}}=\ket{v_{xj}}, 
\qquad \textrm{and} \qquad \ket{v^C_{xj}}=\ket{u_{xj}}.
\end{equation}
Note $\braket{u^C_{xj}}{v^C_{yj}}=(\braket{u_{yj}}{v_{xj}})^*$ and
$(Z_j)_{xy}=(Z_j)_{yx}$ by our symmetric assumption. Since $\sop P$ satisfies the constraints of \cref{def:Gamma2},
\begin{equation}
\sum_j(Z_j)_{xy}\braket{u^C_{xj}}{v^C_{yj}}=\sum_j(Z_j)_{yx}(\braket{u_{yj}}{v_{xj}})^*=\left(\sum_j(Z_j)_{yx}(\braket{u_{yj}}{v_{xj}})\right)^*=B_{yx}^*=B_{xy}.
\end{equation}
Thus the complementary vectors satisfy the same constraints from \cref{def:SP}, and thus produce the same optimal value in
\cref{eq:filteredNormMa}. However, now $w_+(\sop P,x)=w_-(\sop P^C,x)$, and
$w_-(\sop P,x)=w_+(\sop P^C,x)$.
\end{proof}

\CVStransform*

\begin{proof}
For all $x\in X$ and $j\in[n]$, set
\begin{equation}
\ket{u'_{xj}}=\ket{u_{xj}}\sqrt{\frac{\max_x w_-(\sop P,x)}{\max_x w_+(\sop P,x)}}, 
\qquad \ket{v'_{xj}}=\ket{v_{xj}}\sqrt{\frac{\max_x w_+(\sop P,x)}{\max_x w_-(\sop P,x)}}
\end{equation}
It is straightforward to verify that $\sop P'$ is still a converting vector set from $\rho$ to $\sigma$, and that $\max_x w_+(\sop P',x)=\max_x w_-(\sop P',x)$, and that $W(\sop P')$ is the geometric mean of $\max_x w_+(\sop P',x)$ and $\max_x w_-(\sop P',x)$, and so is at most the maximum of either term.
\end{proof}

\section{Proofs from \cref{sec:func} and \cref{sec:StateConv}}\label{app:proofs}

\phaseEstEarly*

\begin{proof}[Proof of \cref{lem:phase_est_early}]
~\\
\textit{Part 1:} Since $f(x)=1$, there is an optimal witness $\ket{w}\in H(x)$
for $x$. Then set $\ket{u}\in \tilde{H}(x)$ to be
$\ket{u}=\alpha\ket{\hat{0}}-\ket{w}.$ Clearly $\Pi_x\ket{u}=\ket{u}$, but also,
$\ket{u}$ is in the kernel of $\Aa{\alpha}$, because
$\Aa{\alpha}\ket{u}=\ket{\tau}-\ket{\tau}=0$.
Thus $\Lal{\alpha}\ket{u}=\ket{u}$, and so $\U{P}{x}{
\alpha}\ket{u}=\ket{u}$; $\ket{u}$ is a 1-valued eigenvector of 
$\U{P}{x}{\alpha}$.

We perform phase estimation on the state
$\ket{\hat{0}}$, so the probability of measuring the state $\ket{0}_B$ in the phase register is \textit{at least} (by \cref{lem:phase_det}), the
overlap of $\ket{\hat{0}}$ and (normalized) $\ket{u}$. This is
\begin{equation}
\frac{|\braket{\hat{0}}{u}|^{2}}{\| \ket{u} \|^2} = \frac{\alpha^2}{\alpha^2+\|\ket{w}\|^2}=\frac{1}{1+\frac{\w{P}{x}}{\alpha^2}}
\end{equation}
Using our assumption that $\w{P}{x}\leq \alpha^2/C$, and a Taylor
series expansion for $C\geq 2$, the probability that we measure the state $\ket{0}_B$ in the phase register is at least $1-1/C$.

~\\ \textit{Part 2:} Since $f(x)=0$, there is an optimal negative witness
$\omega$ for $x$, and we set $\ket{v}\in \tilde{H}$ to be $\ket{v}=\alpha 
(\omega\Aa{\alpha})^\dagger$. By \cref{def:negWit}, $\omega\tau=1$, so $\ket{v}=(\bra{\hat{0}}+\alpha\omega A)^\dagger.$ Again, from \cref{def:negWit},
$\omega A \PHx=0$, so we have $\Pi_x\ket{v}=\ket{\hat{0}}.$

Then when we perform phase estimation of the unitary $\U{P}{x}{\alpha}$ to some
precision $\Theta$ with error $\epsilon$ on state $\ket{\hat{0}}$, by
\cref{lem:phase_det}, we will measure $\ket{0}_B$ in the phase register with probability at most
\begin{equation}\label{eq:Pbound}
\left\|P_{\Theta}\ket{\hat{0}}\right\|^2+\epsilon=\left\|P_{\Theta}\PtHx\ket{v}\right\|^2+\epsilon,
\end{equation}
where throughout this proof, $P_{\Theta}$ is understood to be $P_{\Theta}(\U{P}{x}{\alpha}).$

Now $\ket{v}$ is orthogonal to the kernel of $\Aa{\alpha}$. (To see this, note that if $\ket{k}$ is in the kernel of $\Aa{\alpha}$, then $\braket{v}{k}=\alpha\omega \Aa{\alpha}\ket{k}=0$.)
Applying \cref{spec_gap_lemm}, and setting
$\Theta=\sqrt{\frac{\epsilon}{\alpha^2 W}}$, we have
\begin{equation}\label{eq:negCase1}
\left\|P_{\Theta}\PtHx\ket{v}\right\|^2\leq\frac{\epsilon}{\alpha^2 W}\left\|\ket{v}\right\|^2.
\end{equation}

To bound $\left\|\ket{v}\right\|^2$, we observe that 
\begin{equation}\label{eq:vbound}
\left\|\ket{v}\right\|^2=\left\|\bra{\hat{0}}+\alpha\omega A\right\|^2= 1+\alpha^2w(P,x)\leq 2\alpha^2W.
\end{equation}

Plugging \cref{eq:negCase1,eq:vbound} into \cref{eq:Pbound}, we find that the
probability of measuring $\ket{0}_B$ in the phase register  is at most
$3\epsilon$, as claimed.
\end{proof}

In
 \cref{claim:spec_gap_applied,claim:ampEstbound,claim:stopping,lem:finalStateAnalysis},
 to reduce notation, we use the following simplified conventions. Let
 $\sop{P}=\left(\{\ket{u_{xj}}\},\{\ket{v_{yj}}\}\right)$ be a converting
 vector set from $\rho$ to $\sigma$ and let $W=W(\sop P)$. Let $U=\Up{\sop
 P}{x}{\alpha}{\hat{\varepsilon}}$. We denote $P_\Theta(U)$ and $\Pi_0(U)$ as
 $P_\Theta$ and $\Pi_0$ respectively, where $\Pi_0$ refers to the projection
 onto the subspace that is mapped to the $\ket{0}$ phase register when phase
 estimation on $U$ is run with precision $\hat{\varepsilon}^{3/2}/\sqrt{\alpha W}$
 and accuracy $\hat{\varepsilon}^2$. Likewise $R$ refers to $R(U)$ (Phase
 Reflection) with precision $\hat{\varepsilon}^{3/2}/\sqrt{\alpha W}$, and
 and accuracy $\hat{\varepsilon}^2$. When we write
 $\ket{0}\ket{\rho_x}\ket{0}$, note that the first $\ket{0}$ is a single qubit
 register, while the final $\ket{0}$ is an $O\left(\log\frac{\sqrt{\alpha
 W}}{\hat{\varepsilon}^{3/2}}\log\frac{1}{\hat{\varepsilon}^2}\right)$ qubit
 register. Likewise the final $\ket{0}$ in expressions like
 $\ket{t_{x+}}\ket{0}$ is an $O\left(\log\frac{\sqrt{\alpha
 W}}{\hat{\varepsilon}^{3/2}}\log\frac{1}{\hat{\varepsilon}^2}\right)$ qubit register.

\specBound*

\begin{proof}
 Let $\ket{w} = \sqrt{\frac{\alpha}{\hat{\varepsilon}}} \ket{\psi_{x,\alpha,\hat{\varepsilon}}}$, 
 so $\Lalp \ket{w} = 0$ and $\PHxp \ket{w} = \ket{t_{x-}}$.
Applying \cref{spec_gap_lemm}, we have
 \begin{equation}\label{eq:specGap1}
\| P_\Theta \ket{t_{x-}} \|^2 = \| P_\Theta \PHxp\ket{w} \|^2 \leq \frac{\Theta^2 }{4}\| \ket{w}\|^2.
 \end{equation}
Now
\begin{equation}\label{eq:specGap2}
\left\| \ket{w}\right\|^2 = \frac{\alpha}{\hat{\varepsilon}}\left( \frac{\hat{\varepsilon}}{\alpha}\iner{t_{x-}} - \Sigma_j \iner{u_{xj}}\right) \leq 1 + \frac{\alpha W}{\hat{\varepsilon}^2}.
\end{equation}
Combining \cref{eq:specGap1,eq:specGap2}, and setting $\Theta = \hat{\varepsilon}^{3/2}/\sqrt{\alpha W}$, we have that
\begin{equation}
\| P_\Theta \ket{t_{x-}} \|^2\leq \frac{\hat{\varepsilon}^3}{4\alpha W}\left( 1 + \frac{\alpha W}{\hat{\varepsilon}} \right) \leq \frac{\hat{\varepsilon}^2}{2}.
\end{equation}
\end{proof}

\half*

\begin{proof}
We first prove that
$\|P_0\ket{t_{x+}}\|^2\geq 1-\hat{\varepsilon}$, by following \cite[Claim 4.4]{leeQuantumQueryComplexity2011}. Consider the state
\begin{equation}
\ket{\varphi}=\ket{t_{x+}}+\frac{\sqrt{\hat{\varepsilon}}}{2\sqrt{\alpha}}\frac{2(k-1)}{k}\sum_{j\in[n]}\ket{j}\ket{\nu_{x_j}}\ket{v_{xj}}.
\end{equation}
Note that for all $\ket{\psi_{y,\alpha,\hat{\varepsilon}}}$, because
$\braket{t_{y-}}{t_{x+}}=\frac{1}{2}\left(\braket{\rho_y}{\rho_x}-\braket{\sigma_y}{\sigma_x}\right)$,
and also $\sum_{j:x_j\neq
y_j}\braket{u_{yj}}{v_{xj}}=\braket{\rho_y}{\rho_x}-\braket{\sigma_y}{\sigma_x}$
(see \cref{eq:cvsSum}),
\begin{equation}
\braket{\psi_{y,\alpha,\hat{\varepsilon}}}{\varphi}=\frac{\sqrt{\hat{\varepsilon}}}{\sqrt{\alpha}}\braket{t_{y-}}{t_{x+}}-\frac{\sqrt{\hat{\varepsilon}}}{2\sqrt{\alpha}}\sum_{j:x_j\neq y_j}\braket{u_{yj}}{v_{xj}}=0.
\end{equation}

Because $\ket{\varphi}$ is orthogonal to all of the $\ket{\psi_{y,\alpha,\hat{\varepsilon}}}$, we have
$\Lalp\ket{\varphi}=\ket{\varphi}$. Also, $\PHxp\ket{\varphi}=\ket{\varphi}$ since
$\PHxp\ket{t_{x+}}=\ket{t_{x+}}$ and $\braket{\eta_{x_j}}{\nu_{x_j}}=0$ for every $j$.
Thus $P_0\ket{\varphi}=\ket{\varphi}$. 

Note
\begin{equation}
\braket{\varphi}{\varphi}=1+\frac{\hat{\varepsilon}}{4\alpha}\frac{4(k-1)^2}{k^2}w_+(\sop P,x)< 1+\hat{\varepsilon},
\end{equation}
because of our assumption that $\alpha\geq w_+(\sop P,x)$. Also, $\braket{t_{x+}}{\varphi}=1$,
so 
\begin{equation}
\|P_0\ket{t_{x+}}\|^2> \frac{1}{1+\hat{\varepsilon}}> 1-\hat{\varepsilon}.
\end{equation}
Then by \cref{lem:phase_det}, we have $\|P_0\ket{t_{x+}}\|^2\leq \|\Pi_0\ket{t_{x+}}\ket{0}\|^2$, so 
\begin{equation}\label{eq:lined}
\|\Pi_0\ket{t_{x+}}\ket{0}\|^2> 1-\hat{\varepsilon}.
\end{equation}

Now we turn to analyzing $\|\Pi_0{\ket{0}\ket{\rho_x}}\ket{0}\|$. Using the triangle inequality, we have
\begin{equation}\label{eq:linef}
\|\Pi_0{\ket{0}\ket{\rho_x}}\ket{0}\|\geq\frac{1}{\sqrt{2}}\left(
\|\Pi_0\ket{t_{x+}}\ket{0}\|-\left\|\Pi_0\ket{t_{x-}}\ket{0}\right\|\right).
\end{equation}
Inserting $I=P_\Theta+\overline{P}_\Theta$ with $\Theta = \hat{\varepsilon}^{3/2}/\sqrt{\alpha W}$ into the right-most term and using the triangle inequality again, we have
\begin{equation}\label{eq:linea}
\|\Pi_0\ket{t_{x-}}\ket{0}\|\leq \|\Pi_0(P_\Theta\ket{t_{x-}})\ket{0}\|+ \|\Pi_0(\overline{P}_\Theta\ket{t_{x-}})\ket{0}\|.
\end{equation}
Now
\begin{equation}\label{eq:lineb}
 \|\Pi_0(P_\Theta\ket{t_{x-}})\ket{0}\|\leq \|P_\Theta\ket{t_{x-}}\|\leq \frac{\hat{\varepsilon}}{\sqrt{2}},
 \end{equation} 
 by \cref{claim:spec_gap_applied}. Since we are running Phase Checking to precision $\hat{\varepsilon}^{3/2}/\sqrt{\alpha W}$ and accuracy $\hat{\varepsilon}^2$, from
\cref{lem:phase_det} we get
\begin{equation}\label{eq:linec}
\|\Pi_0(\overline{P}_\Theta\ket{t_{x-}})\ket{0}\|\leq \hat{\varepsilon}.
\end{equation}
Plugging \cref{eq:lined,eq:lineb,eq:linec,eq:linea} into \cref{eq:linef}, we obtain
\begin{equation}
\|\Pi_0{\ket{0}\ket{\rho_x}}\ket{0}\|\geq \frac{1}{\sqrt{2}}\left(\sqrt{1-\hat{\varepsilon}}-\hat{\varepsilon}\left(1+\frac{1}{\sqrt{2}}\right)\right)
\end{equation}
Using a series expansion, we find
\begin{equation}
\|\Pi_0{\ket{0}\ket{\rho_x}}\ket{0}\|^2> \frac{1}{2}-\left(\frac{3}{2}+\frac{1}{\sqrt{2}}\right)\hat{\varepsilon}>\frac{1}{2}-\frac{5}{2}\hat{\varepsilon}.
\end{equation}
\end{proof}

\wrongPhase*

\begin{proof}
Starting from our assumption, writing $\ket{0}\ket{\rho_x}$ in terms of $\ket{t_{x+}}$ and $\ket{t_{x+}}$, and using the triangle inequality, we have
\begin{equation}
\frac{1}{2}-3\hat{\varepsilon}\leq \frac{1}{2}\|\Pi_0(\ket{t_{x+}}+\ket{t_{x-}})\ket{0}\|^2
\leq 
\frac{1}{2}\left(\|\Pi_0\ket{t_{x+}}\ket{0}\|+\|\Pi_0\ket{t_{x-}}\ket{0}\|\right)^2\label{line1}.
\end{equation}
We first bound the term $\|\Pi_0\ket{t_{x-}})\ket{0}\|.$
Inserting the identity operator $I=P_{\Theta}+\overline{P}_{\Theta}$ for $\Theta=\hat{\varepsilon}^{3/2}/\sqrt{\alpha W}$, we have
\begin{align}
 \|\Pi_{0}\ket{t_{x-}}\ket{0}\|
& =\left\|\Pi_{0}\left(\left(P_{\Theta}+\overline{P}_{\Theta}\right)\ket{t_{x-}}\right)\ket{0}\right\|\\
&\leq \|P_{\Theta}\ket{t_{x-}}\|+\|\Pi_{0}\left(\overline{P}_{\Theta}\ket{t_{x-}}\right)\ket{0}\|\\
&\leq \frac{\hat{\varepsilon}}{\sqrt{2}}+\hat{\varepsilon}\leq 2\hat{\varepsilon},
\end{align}
where the second line comes from the triangle inequality and the fact that a projector acting on a vector can not increase its norm.
The first term in the final line comes from \cref{claim:spec_gap_applied}, and the second term   comes \cref{lem:phase_det}.

Plugging into \cref{line1}, we have
\begin{equation}
1-6\hat{\varepsilon}\leq\left(\|\Pi_0\ket{t_{x+}}\ket{0}\|+2\hat{\varepsilon}\right)^2.
\end{equation}
Rearranging, we find:
\begin{equation}
(\sqrt{1-6\hat{\varepsilon}}-2\hat{\varepsilon})^2\leq\|\Pi_0\ket{t_{x+}}\ket{0}\|^2.
\end{equation}
Since $\|\Pi_0\ket{t_{x+}}\ket{0}\|^2+\|\overline{\Pi}_0\ket{t_{x+}}\ket{0}\|^2=1$, we have
\begin{equation}
 \|\overline{\Pi}_0\ket{t_{x+}}\ket{0}\|^2\leq 1-(\sqrt{1-6\hat{\varepsilon}}-2\hat{\varepsilon})^2< 10\hat{\varepsilon}.
\end{equation} 
\end{proof}

\finalState*

\begin{proof}[Proof of \cref{lem:finalStateAnalysis}]
We have 
\begin{align}
\|R\ket{0}\ket{\rho_x}\ket{0}-\ket{1}\ket{\sigma_x}\ket{0}\|=
\frac{1}{\sqrt{2}}\|R(\ket{t_{x+}}+\ket{t_{x-}})\ket{0}-(\ket{t_{x+}}-\ket{t_{x-}})\ket{0}\|\\
\leq \frac{1}{\sqrt{2}}\|(R-I)\ket{t_{x+}}\ket{0}\|+\frac{1}{\sqrt{2}}\|(R+I)\ket{t_{x-}}\ket{0}\|.
\end{align}

In first term, we can replace $R$ with $\Pi_0-\overline{\Pi}_0$ (as described above \cref{lem:phase_refl}), and we can insert $I=\Pi_0+\overline{\Pi}_0$ to get
\begin{equation}
\frac{1}{\sqrt{2}}\|(R-I)\ket{t_{x+}}\ket{0}\|=\frac{1}{\sqrt{2}}\|((\Pi_0-\overline{\Pi}_0)-I)(\Pi_0+\overline{\Pi}_0)\ket{t_{x+}}\ket{0}\|
\end{equation}
Using the fact that $\Pi_0$ and $\overline{\Pi}_0$ are orthogonal, this simplifies to
\begin{equation}
\frac{2}{\sqrt{2}}\|\overline{\Pi}_0\ket{t_{x+}}\ket{0}\|\leq 2\sqrt{5\hat{\varepsilon}},
\end{equation}
by \cref{claim:ampEstbound}.

In the second term, we insert $I=P_\Theta+\overline{P}_\Theta$ and
use the triangle inequality to get 
\begin{align}
\frac{1}{\sqrt{2}}\|(R+I)\ket{t_{x-}}\ket{0}\|
\leq
\frac{1}{\sqrt{2}}\|(R+I)(P_\Theta\ket{t_{x-}})\ket{0}\|+
\frac{1}{\sqrt{2}}\|(R+I)(\overline{P}_\Theta\ket{t_{x-}})\ket{0}\|.
\end{align}
By \cref{lem:phase_refl}, $\frac{1}{\sqrt{2}}\|(R+I)(\overline{P}_\Theta\ket{t_{x-}})\ket{0}\|\leq \frac{\hat{\varepsilon}^2}{\sqrt{2}}$, and
\begin{equation}
\frac{1}{\sqrt{2}}\|(R+I)(P_\Theta\ket{t_{x-}})\ket{0}\|\leq\frac{2}{\sqrt{2}}
\|P_\Theta\ket{t_{x-}}\|\leq \hat{\varepsilon}.
\end{equation}
by a triangle inequality and \cref{claim:spec_gap_applied}.

Combining all of these bounds together, we have
\begin{equation}
\|R\ket{0}\ket{\rho_x}\ket{0}-\ket{1}\ket{\sigma_x}\ket{0}\|
\leq 2\sqrt{5\hat{\varepsilon}}+\hat{\varepsilon}+\frac{\hat{\varepsilon}^2}{\sqrt{2}}<6\sqrt{\hat{\varepsilon}}.
\end{equation}

\end{proof}

\end{document}